\documentclass[12pt]{llncs}
\usepackage[margin=3.2cm]{geometry}
\usepackage{todonotes}

\usepackage{algorithm}
\usepackage{algpseudocode}
\usepackage{multirow}
\usepackage{slashbox}

\usepackage{amsmath,amssymb}
\let\doendproof\endproof
\renewcommand\endproof{~\hfill\qed\doendproof}

\usepackage{psfrag}
\usepackage{graphicx}
\usepackage{amsmath}
\usepackage{amssymb}

\usepackage{subfigure}
\usepackage{xspace}

\usepackage[numbers]{natbib}

\newcommand{\PLS}{\ensuremath{\texttt{PLS}}\xspace}

\newcommand{\NP}{\ensuremath{\texttt{NP}}\xspace}
\newcommand{\coNP}{\ensuremath{\texttt{coNP}}\xspace}
\newcommand{\plsc}{\PLS-complete\xspace}

\title{Computing stable outcomes in symmetric additively-separable hedonic games%
\thanks{This work was supported by EPSRC grants EP/L011018/1 and EP/J019399/1.
This paper combines results from the two conference papers~\cite{GS10} and~\cite{GS11}.}}

\author{
Martin Gairing
\and Rahul Savani
}
\institute{
Department of Computer Science, University of Liverpool.\\  {\tt
\{gairing,rahul.savani\}@liverpool.ac.uk} 
}

\newtheorem{observation}[theorem]{Observation}

\newcommand{\lref}[1]{Lemma~\ref{#1}}
\newcommand{\tref}[1]{Theorem~\ref{#1}}
\newcommand{\fref}[1]{Figure~\ref{#1}}
\newcommand{\pref}[1]{Proposition~\ref{#1}}

\renewcommand{\k}{\kappa}
\newcommand{\nk}{\overline{\kappa}}
\newcommand{\bO}{\ensuremath{\mathcal{O}}}

\newcommand{\ccr}{COMPUTE regime\xspace}
\newcommand{\crr}{RESET regime\xspace}
\newcommand{\gcr}{FIX GATE regime\xspace}
\newcommand{\grr}{RESET GATE regime\xspace}

\newcommand{\is}{\textsc{IS}\xspace}
\newcommand{\cis}{\textsc{CIS}\xspace}
\newcommand{\wcis}{\textsc{sumCIS}\xspace}
\newcommand{\ns}{\textsc{NashStable}\xspace}
\newcommand{\oens}{\textsc{OneEnemyNashStable}\xspace}
\newcommand{\nins}{\textsc{NashStable*}\xspace}
\newcommand{\pa}{\textsc{PartyAffiliation}\xspace}
\newcommand{\oepa}{\textsc{OneEnemy\-Party\-Affi\-liation}\xspace}
\newcommand{\nioepa}{\textsc{One\-Enemy\-Party\-Affi\-liation*}\xspace}
\newcommand{\lmc}{\textsc{LocalMaxCut}\xspace}
\newcommand{\isTwo}{\textsc{2-IS}\xspace}
\newcommand{\isThree}{\textsc{3-IS}\xspace}
\newcommand{\isFour}{\textsc{4-IS}\xspace}
\newcommand{\votein}{\textsc{VoteIn}\xspace}
\newcommand{\voteout}{\textsc{VoteOut}\xspace}
\newcommand{\voteinout}{\textsc{VoteInOut}\xspace}
\newcommand{\vetoin}{\textsc{VetoIn}\xspace}

\newcommand{\eat}[1]{}

\begin{document}

\maketitle

\begin{abstract}
We study the computational complexity of finding stable outcomes in 
hedonic games, which are a class of coalition formation games.
We restrict our attention to symmetric additively-separable hedonic games,
which are a nontrivial subclass of such games that are guaranteed to possess
stable outcomes.
These games are specified by an 
undirected edge-weighted graph: 
nodes are players, an outcome of the game is a partition of 
the nodes into coalitions, and the utility of a node is 
the sum of  incident edge weights in the same coalition.
We consider several 
stability requirements defined in 
the 
literature. These are based on 
restricting feasible player deviations, for example, 
by giving existing coalition members veto power.
We extend these restrictions by considering more general 
forms of preference aggregation for coalition members. 
In particular, we consider 
voting schemes to decide
if coalition members will allow a player to enter or leave their coalition.
For all of the stability requirements we consider, 
the existence of a stable outcome is guaranteed
by a potential function argument, and local improvements will
converge to a stable outcome.
We provide an almost complete characterization of 
these games in terms of the 
tractability of computing such stable outcomes.
Our findings comprise positive results in the 
form of polynomial-time algorithms,
and negative (\PLS-completeness) results.
The negative results extend to more general hedonic games.
\end{abstract}

%


\section{Introduction}

Hedonic games were introduced in the economics literature as a flexible model of
coalition formation~\cite{DrezeGreenberg80a}.  In a hedonic game, each player
has preferences over coalitions and an outcome of the game is a partition of the
players into coalitions.  The defining feature of a hedonic game is that for a
given outcome each player cares only about the other players in the same
coalition.  It is natural to judge the quality of an outcome by how stable it is
with respect to the players' preferences.  Many different notions of stability
appear in the literature.  The survey by Aziz and Savani~\citep{AS14} gives
detailed background on hedonic games and outlines their applications. 
This paper studies the computational complexity of finding stable outcomes in
hedonic games.

In this paper, we consider and extend the stability requirements for hedonic
games introduced in the seminal work of \citet{BoJa02a}.  An outcome of a
hedonic game is called {\em Nash-stable} if no player prefers to be in a
different coalition.  For Nash stability, the feasibility of a deviation depends
only on the preferences of the deviating player.  Less stringent stability
requirements restrict feasible deviations: a coalition may try to hold on to an
attractive player or block the entry of an unattractive player.  In
\citep{BoJa02a}, deviations are restricted by allowing members of a coalition to
``veto'' the entry or exit of a player.
They introduce {\em individual stability}, where every member of a coalition has
a veto that can prevent a player from joining (deviating to) this coalition,
i.e., a player can deviate to another new coalition only if \emph{everyone} in
this new coalition is happy to have her.
They also introduce {\em contractual individual stability}, where, in addition
to a veto for entering, coalition members have a veto to prevent a player from
leaving the coalition - a player can deviate only if everyone in her coalition
is happy for her to leave.
%

The case where every member of a coalition has a veto on allowing players to 
enter and/or leave the coalition can be seen 
as an extreme form of \emph{voting}.
This motivates the study of more general voting mechanisms for allowing players
to enter and leave coalitions.
In this paper, we consider general voting schemes,
for example, where a player is allowed to join a coalition if 
the majority of existing members would like the player to join.
We also consider other methods of \emph{preference aggregation} 
for coalition members. 
For example, a player is allowed to join a coalition only if the
aggregate utility (i.e., the sum of utilities) existing members have for the
entrant is non-negative.
These preference aggregation methods are also considered in
the context of preventing a player from leaving a coalition.
We study the computational complexity of finding stable outcomes
under stability requirements with  
various restrictions on deviations.

\subsection{The model}

In this paper, we study hedonic games with {\em symmetric additively-separable}
utilities, which allow a succinct representation of the game as an {\it
undirected edge-weighted graph} $G=(V,E,w)$. 
For clarity of our voting definitions,  we assume w.l.o.g.~that $w_e\ne0$ for
all $e \in E$ (an edge with weight 0 can be dropped).
Every node $i \in V$ represents a player.  
An outcome is a partition~$p$ of $V$
into coalitions.  
Denote by $p(i)$ the coalition to which $i\in V$ belongs under $p$, and by
$E(p(i))$ the set of edges $\{\{i,j\}\in E \mid j\in p(i)\}$.

The utility of $i\in V$ under~$p$ is the sum of the weights of edges to others in
the same coalition, i.e., 
$$\sum_{e \in E(p(i))} w(e).$$
Each player
wants to maximize her utility, so a player \emph{wants to deviate} if
there exists a (possibly empty) coalition $c$ where
$$
\sum_{e \in E(p(i))} w(e)<\sum_{\{\{i,j\}\in E\ \mid\ j\in c\}} w(\{i,j\}).
$$
We consider different restrictions on player deviations. They restrict 
when players are allowed to join and/or leave coalitions.
A deviation of player $i$ to coalition $c$ is called
\begin{itemize}
\item
\emph{Nash feasible} if player $i$ wants to deviate to $c$.
\item
 \emph{vote-in feasible}  with threshold $T_{in}$ 
if it is Nash feasible and either 
at least a $T_{in}$ fraction of $i$'s edges to $c$ are positive
or $i$ has no edge to $c$. 
\item
\emph{vote-out feasible}  with threshold $T_{out}$ 
if it is Nash-feasible and 
either at least a $T_{out}$ fraction of $i$'s edges to
$p(i)$ are negative or $i$ has no edges within $p(i)$. 
\item
\emph{sum-in feasible}
if it is Nash feasible and 
\begin{align*}
\sum_{\{\{i,j\}\in E\ \mid\ j\in c\}} w(\{i,j\}) \ge 0.
\end{align*}
\item
\emph{sum-out feasible} if it is Nash feasible and 
\begin{align*}
\sum_{e\in E(p(i))} w(e) \le 0.
\end{align*}
\end{itemize}
%
%
Outcomes where no corresponding feasible deviation is possible are called 
\emph{Nash stable, vote-in stable, vote-out stable, sum-in stable}, and \emph{sum-out stable}, respectively.
Outcomes which are vote-in (resp. vote-out) stable with $T_{in}=1$ (resp. $T_{out}=1$) are also called 
\emph{veto-in} (resp. \emph{veto-out}) \emph{stable}.
Note that an outcome is veto-in stable iff it is an \emph{individual stable} outcome; 
and an outcome is both veto-in and
veto-out stable iff it is a \emph{contractual individual stable} outcome (we use the terms
\emph{individual stable} and \emph{contractual individual stable} since they are commonly
used and known in the economics literature following their definition in~\cite{BoJa02a}).


%
\subsection{An example}

\fref{f:example1} gives an example of an additively-separable symmetric hedonic game
that we use to illustrate some of the stability requirements that we have defined.
Consider the outcome $\{\{a,b,d\},\{c,e,f\}\}$. 
The utilities of the players $a,b,c,d,e,f$ are $10,5,-1,5,1,4$, respectively.

\begin{figure}[h]
\begin{center}
\resizebox*{0.3\textwidth}{!}{\includegraphics{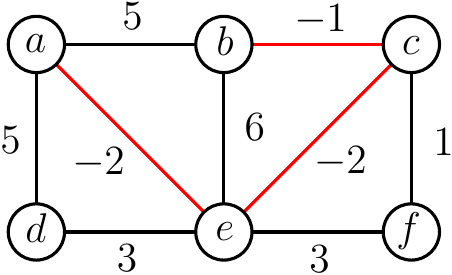}}
\end{center}
\vspace*{-1em}
\caption{An example of an additively-separable symmetric hedonic game.}
\label{f:example1}
\end{figure}

Players $a,b,d,f$ have no Nash-feasible deviations, 
$c$ has a Nash-feasible deviation to go alone
and start a singleton coalition, and $e$ has a Nash-feasible deviation to
join the other coalition.
The deviation of $c$ is not veto-out feasible, since
$f$ prefers $c$ to stay, however it is vote-out feasible
for any $T_{out}\le 0.5$. It is also sum-out feasible.
The deviation of $e$ is not veto-in feasible, but is 
vote-in feasible for any $T_{in}\le 2/3$.
Since there are no deviations that are both veto-in and 
veto-out feasible, this is a contractual individual stable 
outcome.
The outcome $\{\{a,b,d\},\{c\},\{e,f\}\}$ is an individual stable
outcome, and $\{\{a,b,d,e,f\},\{c\}\}$ is Nash stable.

\subsection{Justification of the model}

With the goal of understanding how difficult it is for 
agents to \emph{find} stable outcomes, we 
focus on a model in which they are \emph{guaranteed} to exist. 
The computational complexity of a problem is measured in terms of the
size of its input and therefore depends on the representation of the
problem instance. For games, we desire that the size of the input is
polynomial in the number of players, as this is the natural parameter
with which to measure the size of the game. We consider only such
\emph{succinct representations}, since otherwise 
we can find solutions using trivial
algorithms (enumeration of strategy profiles) 
that are polynomial in the input size.
Our focus on additively-separable games is motivated by
the hardness of even deciding the existence of stable outcomes and 
other solution concepts for more general (universal)
succinct representations, such as hedonic nets~\cite{ElWo09a}.
A \emph{non-symmetric} additively-separable game, 
which is represented by a edge-weighted \emph{directed} graph, 
may not have a Nash-stable outcome~\cite{BoJa02a,Ball04a}, and deciding 
existence is \NP-complete.
We study a more restrictive model where stable outcomes (for all of the stability
requirements we consider) are guaranteed to
exist, noting that our hardness results extend 
to all more general models 
where existence of stable outcomes is either guaranteed or
promised, i.e., instances are
restricted to those possessing stable outcomes.

%
In a \emph{symmetric} additively-separable hedonic game, 
for each of the stability requirements 
we consider, a stable outcome \emph{always exists} by
a simple potential function argument:
the potential function is the total happiness of an outcome, i.e., the sum of players' utilities.
Unilateral player deviations improve the potential.
%
So for all our considered stability requirements, local improvements will find a stable outcome, 
and all the problems we consider are in the complexity class 
\PLS (polynomial local search)~\cite{JPY88}, which we introduce next.

\subsection{Local search and the complexity class \PLS}

Local search is one of few general and successful approaches to difficult
combinatorial optimisation problems. A local search algorithm tries to find an
improved solution in the \emph{neighborhood} of the current solution. A solution
is \emph{locally optimal} if there is no better solution in its neighborhood.
\citet{JPY88} introduced  the complexity class \PLS (polynomial local search) to
capture those local search problems for which a better neighboring solution can
be found in polynomial time if one exists, and a local optimum can be verified
in polynomial time.

A problem in \PLS comprises a finite set of candidate solutions.
Every candidate solution has an associated non-negative integer cost, and a 
neighbourhood of candidate solutions.
In addition, a \PLS problem is specified by the following three polynomial-time 
algorithms that:
\begin{enumerate}
\item construct an initial candidate solution;
\item compute the cost of any candidate solution in polynomial time;
\item given a candidate solution, provide a neighbouring solution with
	lower cost if one exists.
\end{enumerate}
The goal in a \PLS problem is to find a local optimum, that is, a candidate
solution whose cost is no more than the cost of any of its neighbours.

Suppose $A$ and $B$ are problems in \PLS. Then $A$ is \PLS-reducible to $B$ if
there exist polynomial-time computable functions $f$ and $g$ such that $f$ maps
an instance $I$ of $A$ to an instance $f(I)$ of $B$, and $g$ maps the local optima
of instance $f(I)$ of $B$ to local optima of instance $I$. 
A problem in \PLS is \PLS-complete if all problems in \PLS are \PLS-reducible to it.
Prominent \PLS-complete problems include finding a locally optimal max-cut in a graph
(\lmc) \citep{SY91}, or a stable solution in a Hopfield network \citep{JPY88}.
\PLS captures the problem of finding pure Nash equilibria for many classes of
games where pure Nash equilibria are guaranteed to exist, such as congestion
games~\cite{FPT04}, for which is also \plsc to find a pure equilibrium.

On the one hand, finding a locally optimal solution is presumably easier than
finding a global optimum; in fact, it is very unlikely that a \PLS problem is
\NP-hard since this would imply \NP$=$\coNP \citep{JPY88}. On the other hand, a
polynomial-time algorithm for a \PLS-complete problem would resolve  a number of
long open problems, e.g., since it would show that \emph{simple stochastic
games} can be solved in polynomial time~\citep{Yan08}. Thus, \PLS-complete
problems are believed not to admit polynomial-time algorithms.

\subsection{Computational problems}

We define the search problems, \ns, \is (individual stable), \cis (contractual
individual stable), \votein, and \voteout of finding a stable outcome for the
respective stability requirement.  
We introduce \voteinout as the search problem of finding an outcome which is
vote-in and vote-out stable.  
All voting problems are parametrized by $T_{in}$ and/or $T_{out}$.  
We recall that outcomes which are vote-in (resp. vote-out) stable with
$T_{in}=1$ (resp. $T_{out}=1$) are also called \emph{veto-in} (resp.
\emph{veto-out}) \emph{stable}, so \is is the computational problem of finding a
veto-in stable outcome, and \cis is the problem of finding an outcome that is
both veto-in and veto-out stable.
We also introduce \wcis as the problem of finding an outcome which is sum-in and
sum-out stable.

Symmetric additively-separable hedonic games are closely related to 
party affiliation games, which are also specified by an undirected edge-weighted graph.
In a party affiliation game each player must choose between one of 
two ``parties''; a player's happiness is the sum of her edges to nodes 
in the same party; in a stable outcome no player would prefer to be in 
the other party.
The problem \pa is to find a stable outcome in such a game.
If such an instance has only negative edges then it is equivalent to the problem 
\lmc, which is to find a stable outcome of a local max-cut game.
In party affiliation games there are at most two coalitions, while  
in hedonic games any number of coalitions is allowed.
Thus, whereas \pa for instances with only negative 
edges is \PLS-complete~\cite{SY91}, 
\ns is trivial in this case, as the outcome where all players are in singleton coalitions is Nash-stable. 
Both problems are trivial when all edges are non-negative, in which case the 
 grand coalition of all players is Nash-stable.
Thus, interesting hedonic games contain both positive and negative edges.

The problem \oepa is to find a stable outcome of a party
affiliation game where each node is incident to at most one
negative edge.
In this paper, we use a variant of this problem as a starting point for
some of our reductions: 
\begin{definition}
We define the problem \nioepa 
as a restricted version of \oepa which is restricted to instances
where no player is ever indifferent between the two coalitions.
\end{definition}
We show in 
Corollary~\ref{c:oepastar} (page \pageref{c:oepastar}) that 
\nioepa is \PLS-complete. 

\subsection{Our results}

In this paper, we examine the complexity of computing stable outcomes in symmetric 
additively-separable hedonic games. 
We observe that \ns, i.e., the problem of computing a Nash-stable outcome, is 
\PLS-complete (Observation~\ref{o:NashPLS}). 
Here, we give a simple reduction from \pa, which was shown to be \PLS-complete
in~\cite{SY91}.
Our reduction relies on a method to ensure that all stable outcomes use exactly 
two coalitions (where in general there can be as many coalitions as players).

Moreover, we study \is, i.e., the problem of finding an individually-stable outcome.
We show that if the outcome is restricted to contain at most two coalitions, 
an individually-stable outcome can be found in polynomial time (Proposition~\ref{p:is2}).
This suggests that a reduction showing \PLS-hardness for IS cannot be as simple 
as for \ns: one would need to
construct hedonic games that allow three or more coalitions.

In order to prove that \is is \plsc, we first define a restricted version of
\pa, called \oepa, in which each player dislikes at most one other player. Our
main technical result is that \oepa  is \PLS-complete
(Theorem~\ref{t:main_pls}). This reduction is from \textsc{CircuitFlip} and is
rather involved. The instances of \oepa that we produce via this reduction have
the property that no player is ever indifferent between two coalitions. We then
show that such instances can be reduced to \is.

Perhaps surprisingly, given the apparently restrictive nature of the stability
requirement, we show that \wcis is \PLS-complete (Theorem~\ref{thm:wcis}).

In contrast, we show that the problem \cis of finding a
contractually-individually-stable outcome can be solved in polynomial time. We
make explicit two conditions in Propositions~\ref{prop:existsneg}
and~\ref{prop:justgrow}, both met in the case of \cis, that (individually)
guarantee that local improvements converge in polynomial time. We use these
propositions to give further positive results for other combinations of 
restrictions, where either the entering or leaving restriction is veto based.

\begin{figure*}
\begin{center}
\resizebox*{.9\textwidth}{!}{%
\begin{tabular}{|c||c|c|c|c|}
\hline
\multirow{3}{*}{\backslashbox{Leave}{Enter}} 
& 1:
& 2:
& 3:
& 4: 
\\
& \multirow{2}{*}{no restr.} 
& \multirow{2}{*}{sum-in} 
& \multirow{2}{*}{veto-in} 
& \multirow{2}{*}{vote-in} 
\\
& & & & \\
\hline
\hline
A:& \ns  &  & \is & \votein \\
\multirow{2}{*}{no restr.} 
& \PLS-complete & \PLS-complete & \PLS-complete & \PLS-complete \\
& Observation \ref{o:NashPLS} &   Observation \ref{o:NashPLS}      & Theorem~\ref{thm:is} & \tref{thm:votein}\\
\hline
B:&       & \wcis &   &   \\
\multirow{2}{*}{sum-out}
& \plsc & \plsc & P & ? \\
&   \tref{thm:wcis}     & \tref{thm:wcis} & \pref{prop:existsneg} &   \\
\hline
C:&       &       &   \cis &   \\
\multirow{2}{*}{veto-out} 
&   P    &   P    &   P &  P  \\
& \pref{prop:justgrow} & \pref{prop:justgrow} & \pref{prop:existsneg} or \ref{prop:justgrow}& \pref{prop:justgrow} \\
\hline
D:& \voteout           &       &   & \voteinout  \\
\multirow{2}{*}{vote-out} 
& ?  & ?  & P  &  P \tiny{($T_{in},T_{out}> 0.5$)} \\
& (see \tref{thm:voteout}) &  (see \tref{thm:voteout}) & \pref{prop:existsneg} & \tref{thm:voteinout} \\
\hline
\end{tabular}
}
\end{center}
\caption{Table showing the computational complexity of the search problems 
for different entering and leaving deviation restrictions.
Note that columns 1 and 2 are essentially equivalent, since
if a player has a Nash-feasible deviation that results in a negative
payoff, she also has a sum-in feasible (and hence also
Nash-feasible) deviation, namely to form a singleton coalition.
}
\label{f:table}
\end{figure*}

Finally, we study the complexity of finding vote-in and vote-out stable
outcomes. Using a different argument to the polynomial-time cases mentioned
previously, we show that local improvements converge in polynomial time in the
case of vote-in- and vote-out- stability with $T_{in},T_{out}\ge 0.5$ and
$T_{in}+T_{out}>1$ (Theorem~\ref{thm:voteinout}). We show that if we require
vote-in-stability alone, we get a \plsc search problem
(Theorem~\ref{thm:votein}). The problem of finding a vote-out stable outcome is
conceptually different (e.g., we can find a veto-out-stable outcome in polynomial
time, whereas it is \plsc to find a veto-in-stable outcome). The technical
difficulty in proving a hardness result for \voteout is restricting the number
of coalitions. Ultimately, we leave the complexity of \voteout open, but do show
that $k$-\voteout, which is the problem of computing a vote-out stable outcome
when at most $k$ coalitions are allowed, is \plsc (\tref{thm:voteout}). 

Our results are summarized in~\fref{f:table}, which gives an almost complete
characterization of tractability.

\subsection{Related work}

Hedonic coalition formation games were first considered by \citet{DrezeGreenberg80a}.
\citet{Gree94a} later surveyed coalition structures in game theory and
 economics. 
Based on \cite{DrezeGreenberg80a}, 
\citet{BoJa02a} formulated different stability concepts in
the context of hedonic games - see also the survey~\cite{SD07}.
These stability concepts were our motivation to introduce
definitions of stability based on voting and aggregation.
  
The general focus in the game theory community has been on
characterizing the conditions for which stable outcomes
exist. 
%
%
 \citet{BuZw03a} showed that additively-separable and symmetric
preferences guarantee the existence of a Nash-stable outcome. 
They also showed that 
under certain different conditions on the preferences,
the set of Nash-stable outcomes can be empty but the set of individually-stable 
partitions is always non-empty. 

\citet{Ball04a} showed that for hedonic games represented by an
\emph{individually rational list of coalitions}, the complexity of checking
whether core-stable, Nash-stable or individual-stable outcomes exist is
\NP-complete, and that every hedonic game has a
contractually-individually-stable solution.  \citet{SuDi10a} showed
that for additively-separable hedonic games checking whether a core-stable,
strict-core-stable, Nash-stable or individually-stable outcome exists is
\NP-hard. For core-stable and strict-core-stable outcomes those \NP-hardness
results have been extended by  \citet{ABS11} to the case of symmetric player
preferences.  Recently, the paper \cite{PE15} unifies and extends several of
these results by identifying simple conditions on expressivity of hedonic games
that are sufficient for the problem of checking whether a given game admits a
stable outcome to be computationally hard.

\citet{BL09} studied the tradeoff between stability and social welfare in additively-separable hedonic games.
\citet{ElWo09a} characterize the complexity of problems 
related to coalitional stability for hedonic games represented by
 hedonic nets, a succinct, rule-based representation
 based on marginal contribution nets (introduced by~\citet{IS05}).
\citet{Cech08a} surveys algorithmic problems related to stable outcomes.

%
%

The definition of party affiliation games we use appears in~\citet{BBM09}.
Recent work on local max cut and party affiliation games has focused on
approximation~\cite{BCK10,CMS06}; see also \cite{OPS04}.  For surveys on the
computational complexity of local search, see \cite{MDT10,AL97}.
Our \PLS-hardness results use ideas from \citet{Kre89,SY91,MT10}, and in
particular \citet{ET11}. 
We use the \PLS-completeness of \lmc which was shown in~\citet{SY91}.

There is an extensive literature on weighted voting games, which are formally
simple coalitional games.  For such a game, a ``solution'' 
is typically a vector (or set of vectors) of payoffs for the players, rather
than a coalition structure as in our setting; for recent work on computational
problems associated with weighted voting games see~\cite{EP09,EGGW09}. 
\citet{DePa94a} examined the computational complexity of computing solutions for
coalitional games for a model similar to additively-separable hedonic games,
where the game is given by an edge-weighted graph, and the value of a coalition
of nodes is the sum of weights of edges in the corresponding subgraph.
%


\subsection{Outline of the paper}

In Section~\ref{sec:nash}, we show that \ns is \plsc.
In Section~\ref{sec:oepa}, we prove our main technical result: \oepa is 
\plsc.
\oepa is the starting point for our reduction to \is (i.e., \vetoin) in Section~\ref{sec:is}, 
which shows that it too is \plsc.
In Section~\ref{sec:other_veto}, we show that the remaining veto-based problems,
namely all those (except for \is) in row C and column $3$ in Figure~\ref{f:table}, 
can be solved in polynomial time.
In Section~\ref{sec:sum_cis}, we show that the problem \wcis is \plsc.
In Section~\ref{sec:voting}, we give both positive and negative results for
computing stable outcomes under various voting-based stability requirements.
Finally, in Section~\ref{sec:conc}, we conclude with open problems.

\section{Nash stability and restricting the number of coalitions}
\label{sec:nash}

In this section, we show that \ns is \PLS-complete via a reduction from \pa.
Recall that in Nash-stable outcomes for hedonic additive separable games,
players might form more than two coalitions while party affiliation games are
restricted to two coalitions.
To deal with this we use a mechanism, called \emph{supernodes},
which can restrict the number of coalitions that will be non-empty in stable
outcomes.
In the reduction in this section we use two supernodes to restrict to two
coalitions; later in the paper we will use a variable number of $k$ supernodes
to restrict the number of coalitions to $k$.

\begin{observation}
\label{o:NashPLS}
  \ns is \PLS-complete.
\end{observation}

\begin{proof}
  Consider an instance of \pa, represented as an edge-weighted graph $G=(V,E,w)$. 
  We augment $G$ by introducing two 
  new players, called \emph{supernodes}. 
  Every player $i\in V$ has an edge of weight $W>\sum_{e\in E} |w_e|$ to each of the supernodes. 
  The two supernodes
  are connected by an edge of weight $-M$, where $M> |V|\cdot W$. 
  By the choice of $M$ the two supernodes will be in different coalitions in any Nash-stable outcome of the resulting hedonic game. 
  Moreover, by the choice of~$W$, each player will be in a coalition with one of the supernodes. 
  So, in every Nash-stable outcome we have exactly two coalitions. 
  The fact that edges to supernodes have all the same weight directly implies a one-to-one
  correspondence between the Nash-stable outcomes in the hedonic game and in the party affiliation game.
\end{proof}

\section{Key technical result: \textnormal{\oepa} is \PLS-complete}
\label{sec:oepa}

In this section, we prove our key technical result, that \oepa is \PLS-complete.
The instances that we produce have the useful property that no player is every
indifferent between two coalitions, which we make explicit in
Corollary~\ref{c:oepastar}.
We use these special instances for other reductions in this paper.

The starting point for our reduction to \oepa is the prototypical \plsc problem {\sc
CircuitFlip}, introduced and shown to be \plsc in~\citep{JPY88}.

\begin{definition}
An instance of {\sc CircuitFlip} is a boolean circuit with $n$ inputs and $n$
outputs.  A feasible solution is an assignment to the inputs and the value of a
solution is the output treated as a binary number.  The neighbourhood of an
assignment consists of all assignments obtained by flipping exactly one input
bit.  The objective is to maximise the value.
\end{definition}

%

\begin{theorem}
\label{t:main_pls}
\oepa is \PLS-complete.
\end{theorem}

\begin{proof}
We reduce from {\sc CircuitFlip}. Let $C$ be an instance of  {\sc CircuitFlip} with inputs $V_1, \ldots, V_n$, outputs $C_1, \ldots, C_n$, and gates
$G_1, \ldots, G_N$. We make the following simplifying assumptions about $C$: 
(i) The gates are topologically ordered so that if the output of $G_i$ is an input to $G_j$ then $i>j$. 
(ii) All gates are NOR gates with fan-in 2.  
(iii) $G_1, \ldots, G_n$ is the output and $G_{n+1}, \ldots, G_{2n}$ is the (bitwise) negated output of $C$ with $G_1$ and $G_{n+1}$ being the most significant bits.  
(iv) $G_{2n+1}, \dots, G_{3n}$ outputs a (canonical) better neighbouring solution if $V_1, \ldots, V_n$ is not locally optimal. 

We use two complete copies of $C$. One of them represents the current solution while the other ones represents the next (better) solution.
Each copy gives rise to a graph. We will start by describing our construction for one of the two copies and later show how they interact.
Given $C$ construct a graph $G_C$ as follows:

We have nodes $v_1, \ldots, v_n$ representing the inputs of $C$,
and nodes $g_i$ representing the output of the gates of $C$. We will also use $g_i$ to refer to the whole gate.
For $i\in[n]$, denote by $w_i:=g_{2n+i}$ the nodes representing the better neighbouring solution.  
Recall that $g_1, \ldots, g_n$ represent the output of $C$ while $g_{n+1}, \ldots, g_{2n}$ correspond to the negated output.


In our party affiliation game we use $0$ and $1$ to denote the two coalitions. We slightly abuse notation
by using $u=\k$ for $\k\in\{0,1\}$ to denote that node $u$ is in coalition $\k$. 
In the construction, we assume the existence of nodes with a fixed coalition. 
This can be achieved as in the proof of Observation~\ref{o:NashPLS}
with the help of supermodels.
We use $0$ and $1$ to refer to those constant nodes. 
In the graphical representation (cf. \fref{f:mainfig}), we represent those constants by square nodes. 

We follow the exposition of \citet{SY91} and \citet{ET11} and use types to introduce our construction. Nodes may be part of multiple types. 
In general types are ordered w.r.t. decreasing edge weights. So earlier types are more important. Different types will serve different purposes.

\paragraph{Type 1: Check Gates.}
For each gate $g_i$ we have a three-part component as depicted in Figure~\ref{f:gate}. 
The inputs of $g_i$, denoted $I_1(g_i)$ and $I_2(g_i)$, 
are either inputs of the circuit or outputs of some gate with larger index. The main purpose of this component is to check if 
$g_i$ is correct, i.e., $g_i=\neg (I_1(g_i)\wedge I_2(g_i))$, and to set $z_i=1$ if $g_i$ is incorrect.
The $\alpha$, $\beta$, $\gamma$, $\delta$ and $\lambda$ nodes are local nodes for the gate.
A gate can be in two operational modes, called \emph{gate push regimes}. Type 7 will determine 
in which of the following push regimes a gate is.

\begin{definition}[Gate push regimes] 
In the \emph{\grr } $\alpha_{i,1}$, $\alpha_{i,2}$, $\gamma_{i,1}$ and $\gamma_{i,2}$ get a bias towards $1$ while $\lambda_{i,1}, \lambda_{i,2}, \beta_{i,1}, \beta_{i,2}, \beta_{i,3}, \delta_{i,1}, \delta_{i,2}$ and $\gamma_{i,3}$ get a bias towards $0$. In the \emph{\gcr } we have opposite biases. 
\end{definition}
%

\paragraph{Type 2: Propagate Flags.} 
In order to propagate incorrect values for the $z$ variables we interconnect them as in Figure~\ref{f:cata} by using the topological order on the gates.  
Observe that for any locally optimal solution $z_i=1$ enforces $z_j=1$ for all $j<i$. The component is also used to (help to) fix the gates
in order and to RESET them in the opposite order.
Node $z_{N+1}$ is for technical convenience.

\medskip
Type 1 and 2 components are the same for both copies. In the
following we describe how the copies interact. 
We denote the two copies of $C$ by $C^0$ and $C^1$ and also
use superscripts to distinguish between them for nodes of
type 1 and 2.

\paragraph{Type 3: Set/Reset Circuits.}
The component of type 3 interconnects the $z$-flags from the two circuits $C^0, C^1$. 
This component is depicted in Figure \ref{f:type3} and has multiple purposes. 
First, it ensures that in a local optimum $d^0$ and $d^1$ are not  both $1$. 
Second, at the appropriate time, it triggers to reset the circuit with smaller output. 
And third, it locks
$d^0$ or $d^1$ to 1 and resets them back to $0$ at the appropriate times. 

\medskip
The $z$ and $y$ nodes can also be in two different operational modes called \ccr and \crr which is determined by Type 6.
\begin{definition}[Circuit push regimes]
Let $\k\in\{0,1\}$.
In the \emph{\ccr for $z^\k$} all $z^\k_i$ get a bias to $0$ for all $0\le i \le N+1$ and $y^\k$ gets a bias
to~$1$.  In the \emph{\crr for $z^\k$} we give opposite biases.
\end{definition}

\paragraph{Type 4: Check Outputs.}
This component compares the current output of the two circuits and gives incentive to set one of the nodes $d^0$ or $d^1$ to 1 for which the output of the corresponding 
circuit is smaller.
For all $i\in[n]$, we have edges $(d^0, g^0_{n+i}), (d^0, g^1_{i}), (d^1,g^1_{n+i})$, $(d^1, g^0_{i})$ and $(0, g^0_{n+i}), (1, g^1_{i}), (0,g^1_{n+i}), (1, g^0_{i})$ 
of weight $2^{2n+1-i}$. 
To break symmetry we have edges  $(0,d^0), (1,d^1)$ of weight $2^n$. 

\begin{figure*}
\centering
\subfigure[Type 1. Extra factor: $2^{2n+5i}$]{
\label{f:gate}
\resizebox*{0.28\textwidth}{!}{\includegraphics{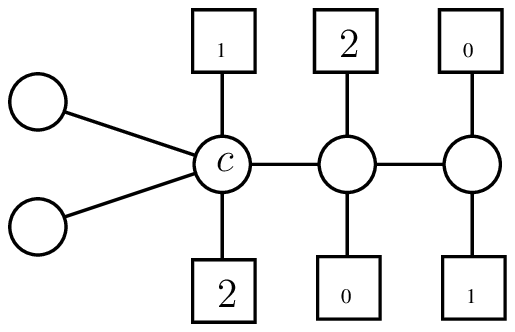}}
\hskip0.5cm
\resizebox*{0.28\textwidth}{!}{\includegraphics{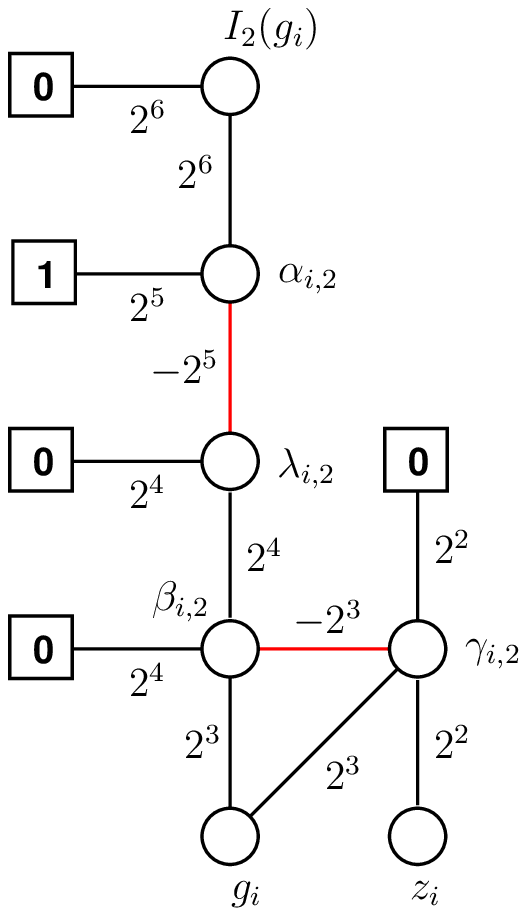}}
\hskip0.5cm
\resizebox*{0.35\textwidth}{!}{\includegraphics{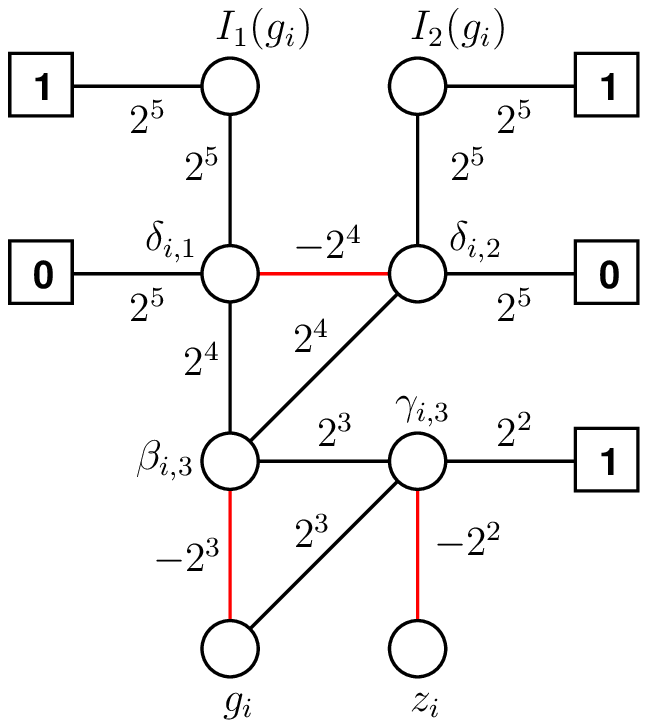}}
}
\subfigure[Type 2. Extra factor: $2^{2n}$]{\resizebox*{\textwidth}{!}{\includegraphics{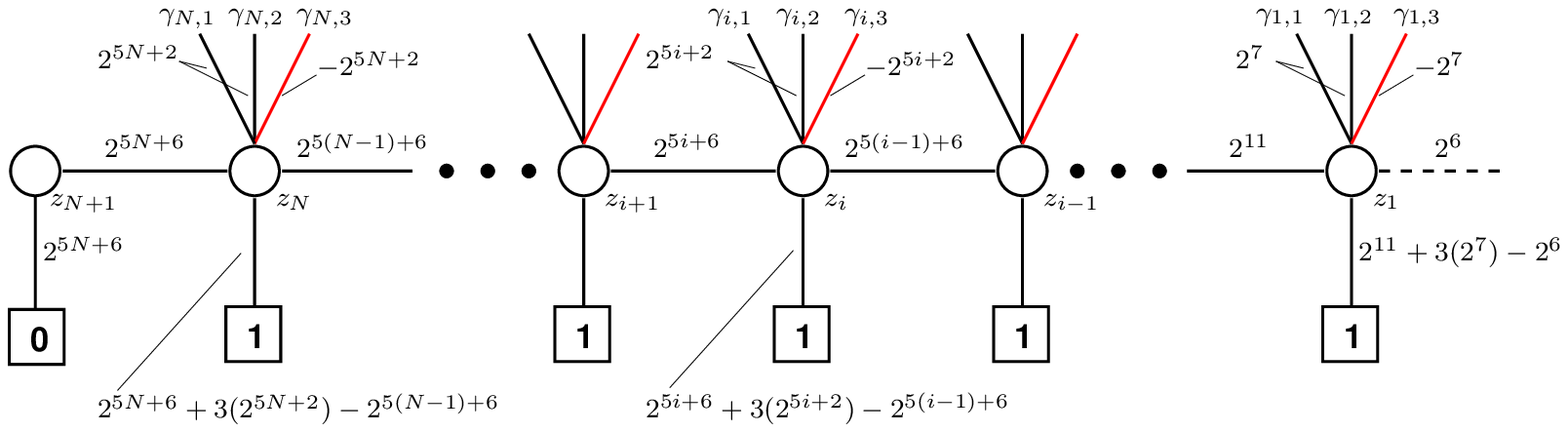}\label{f:cata}}}
\subfigure[Type 3. Extra factor: $2^{2n}$]{\label{f:type3}\resizebox*{0.45\textwidth}{!}{\includegraphics{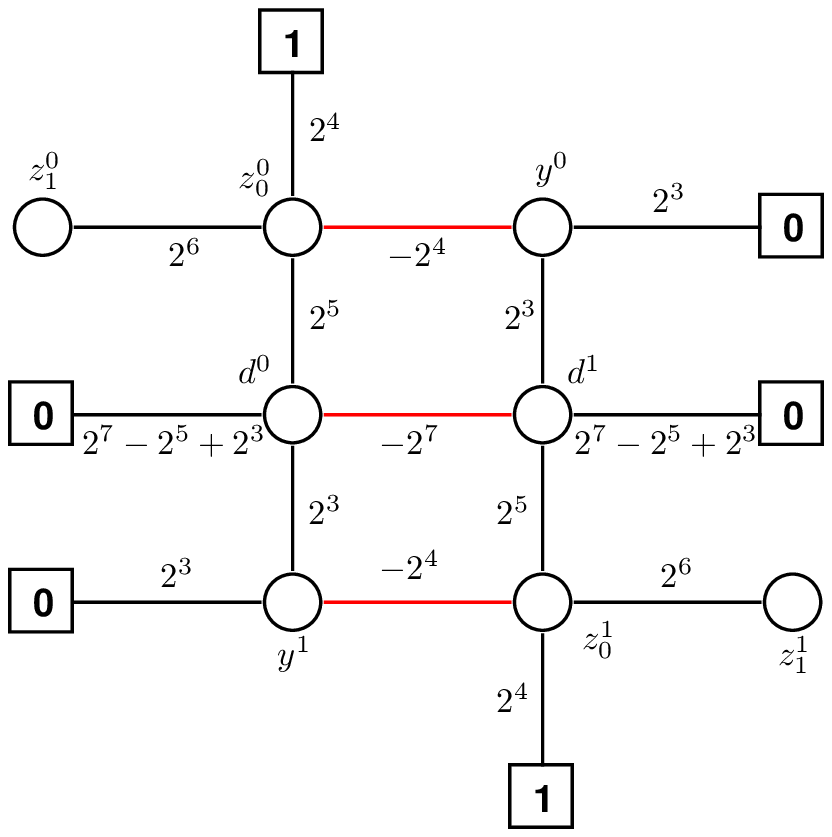}}}
\hskip1.5cm
\subfigure[Type 5. No extra factor]{\label{f:copy}\resizebox*{0.32\textwidth}{!}{\includegraphics{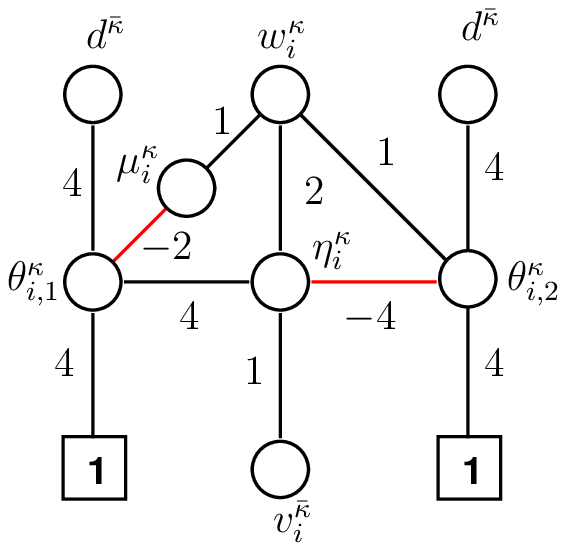}}}
\caption{
Components of type 1,2,3, and 5. Edge weights have to be
multiplied by the factors given above.
\label{f:mainfig}
}
\end{figure*}

\paragraph{Type 5: Feedback Better Solution.}
This component is depicted in Figure \ref{f:copy}. 
It is used to feedback the improving solution of one circuit to the input of the other circuit. 
Its operation is explained in \lref{l:copy_better}.

\medskip
For the remaining types we use 
the following lemma and definition which are analogous to those in \cite{ET11,MT10}.
\begin{lemma}
\label{l:compute_function(old)}
For any  polynomial-time computable function $f:\{0,1\}^k\mapsto \{0,1\}^m$ one can construct a graph $G_f(V_f,E_f, w)$ having the following properties:
(i) there exist $s_1, \ldots, s_k$, $t_1, \dots, t_m \in V_f$ with no negative incident edge,
(ii) each node in $V_f$ is only incident to at most one negative edge, 
(iii) $f(s)=t$ in any Nash-stable solution of the party affiliation game defined by $G_f$. 
\end{lemma}

\begin{proof}
It is well known that for any polynomial computable function $f:\{0,1\}^k\mapsto \{0,1\}^m$
one can construct a circuit $C$ with polynomial many gates that implements this function \cite[Theorem 9.30]{Sip06}. 
Clearly, we can also restrict $C$ to NOR gates with fan-in and fan-out at most 2. 
Organize the gates in levels according to their distance to $C$'s output; output gates are at level 1.   

We replace each gate $g_i$ at level $\ell$ with the gadget below.
Nodes $a,b$ are inputs and $d$ is the output of the gate. 

\begin{figure}[h]
\centering
\resizebox*{0.314\textwidth}{!}{\includegraphics{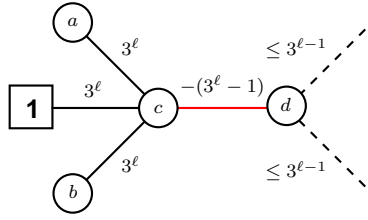}}
\label{f:nor}
\caption{NOR gate}
\end{figure}

If $a$ (or $b$) is an input of the circuit then we connect $a$ to the corresponding input $s$-node by an edge of weight $3^{\ell+1}$.  
If $\ell=1$, i.e., $g_i$ is an output gate, then we connect $d$ to the corresponding output $t$-node with an edge of weight~$1$. 
Otherwise ($\ell>1$), $d$ is also the input to at most $2$ lower level gates. The corresponding edges have weight at most $3^{\ell-1}$.  
In any Nash-stable solution, $d=1$ if and only if $a=b=0$. In 
other words $d=NOR(a,b)$. 
The claim follows since our construction fulfils properties (i), (ii) and (iii).
\end{proof}


\begin{definition}
\label{d:function(old)}
For a  polynomial-time computable function $f:\{0,1\}^k\mapsto \{0,1\}^m$ we say that $G_f$ 
as constructed in \lref{l:compute_function(old)} is a graph that \emph{looks} at $s_1, \dots, s_k \in V_f$ and \emph{biases} $t_1, \dots, t_m \in V_f$
according to the function $f$.
\end{definition}
In the final three types we look at and bias nodes from the lower types already defined. 
For the final types we do not give explicit edge weights. 
In order that the ``looking'' has no side-effects 
on the operation of the lower types, we scale edge weights in these types such that any edge weight of lower type is larger 
than the sum of the edge weights of all higher types. 
More precisely,
for $j\in\{5,6,7\}$,
the weight of the smallest edge of type $j$ is larger than the sum of weights of all edges of types $(j+1), \dots, 8$.

In the following, denote by $C(v)$ the value of circuit $C$ of the \textsc{CircuitFlip} instance on input $v=(v_i)_{i\in [n]}$ and $w(v)$ the better neighbouring solution. 
Both are functions as in Definition \ref{d:function(old)}.

\paragraph{Type 6: Change Push Regimes for z.}
The component of type 6 looks at $v^0$, $v^1$, $d^0$, $d^1$, $\eta^0$ and $\eta^1$ (type 5) and biases $z_i^0, z_i^1, y^0$ and $y^1$ as follows.
$z^0$ is put in the \ccr if at least one of the following 3 conditions is fulfilled: (i) $C(v^0)\ge C(v^1)$, (ii) $w(v^1)=v^0$, or (iii) $w(v^1)\neq \eta^1 \wedge d^0=1$. Else $z^0$ is put into the \crr. 
Likewise 
$z^1$ is put in the \ccr if at least one of the following three conditions is fulfilled: (i) $C(v^0)< C(v^1)$, (ii) $w(v^0)=v^1$, or (iii) $w(v^0)\neq \eta^0 \wedge d^1=1$. Else $z^1$ is put into the \crr. 
Note that conditions (i) and (ii) are important for normal computation, while (iii) is needed to overcome bad starting configurations.

\paragraph{Type 7: Change Push Regimes for Gates.}
For each $i\in [N]$ and $\kappa\in\{0,1\}$, if $z_{i+1}^\kappa=0$ we put the local variable of $g_i^\kappa$ in the \gcr and in the \grr otherwise.
\paragraph{Type 8: Fix Incorrect Gate.}
For each $i\in [N]$ and $\kappa\in\{0,1\}$, the components of type 8 give a tiny offset to $g_i^\kappa$ for computing correctly. 
For each gate $g_i^\kappa$ we look at  $\alpha^{\k}_{i,1},\alpha^{\k}_{i,2}$ and bias  $g_i^\kappa$ to $\neg (\alpha^{\k}_{i,1}\wedge \alpha^{\k}_{i,2})$.

\medskip
This completes our construction. We proceed by showing properties of Nash-stable outcomes. 
Each of the following six lemmas should be read with the implicit clause: 
``\emph{In every Nash-stable outcome.}''

\begin{lemma}
\label{l:copy_better}
Let $\k\in\{0,1\}$, then the following holds for all $i\in[n]$:\\
(a) If $d^{\nk}=0$ then $w_i^{\k}$ is indifferent w.r.t. edges of type 5.  \\
(b)
If $d^{\nk}=1$ then $\eta_i^\k=w_i^\k$.
\end{lemma}

\begin{proof}
Suppose $d^{\nk}=0$. If $\eta^\kappa_{i}=0$ then 
$\theta^\kappa_{i,1}=1$ and $\theta^\kappa_{i,2}=\mu^\kappa_{i}=0$,  
and hence $w_i^{\k}$ is indifferent w.r.t. edges of type 5.
The case $\eta^\kappa_{i}=1$ is symmetric.
This proves~(a).
If $d^{\nk}=1$ then $\theta^\kappa_{i,1}=\theta^\kappa_{i,2}=1$, so  
$\eta_i^\k$ is indifferent w.r.t. the edges connecting it to 
$\theta^\kappa_{i,1}$ and $\theta^\kappa_{i,2}=1$. 
Hence $\eta_i^\k$ will copy the value of~$w_i^\k$, which proves (b).
\end{proof}

\begin{lemma}
\label{l:gate_incorrect}
If $g^\k_i$ is incorrect then $z_i^\k=1$. 
If $z_i^\k=1$ then $z_j^\k=1$ for all $0\le j\le i$ and $y^\k=0$.
\end{lemma}

\begin{proof}
Gate $g^\k_i$ can be incorrect in two ways:
\begin{itemize}
\item[(i)] $I_j(g_i^\k)=1$ for some $j\in\{1,2\}$ and $g^\k_i=1$,
\item[(ii)] $I_1(g_i^\k)=I_2(g_i^\k)=0$ and $g^\k_i=0$.  
\end{itemize}
For case (i) observe that $I_j(g_i^\k)=1 \implies \alpha^\k_{i,j}=1 \implies \lambda_{i,j}=0 \implies \beta^\k_{i,j}=0$. Together with $g^\k_i=1$ this directly implies $\gamma^\k_{i,j}=1$.
Consider now case (ii).  Since $I_1(g_i^\k)=I_2(g_i^\k)=0$ we have $\delta^\k_{i,1}=\delta^\k_{i,2}=0$ and therefore $\beta^\k_{i,3}=0$. Together with $g_i^\k=0$ this directly implies $\gamma^\k_{i,3}=0$.
In either case, this implies $z_i^\k=1$, proving the first part of the lemma.

The second claim holds by induction, since $z_i^\k=1$ enforces $z_{i-1}^\k=1$, while $z_0^\k=1$ enforces $y^\k=0$. 
\end{proof}

\begin{lemma}
\label{l:inputs_indifferent}
If $z_{i+1}^\k=1$ then the inputs $I_1(g_i^\k)$ and  $I_2(g_i^\k)$ are indifferent with
respect to the type 1 edges of gate $g_i^\k$.
\end{lemma}

\begin{proof}
We show that $\alpha^\k_{i,1}=\alpha^\k_{i,2}=1$ and $\delta^\k_{i,1}=\delta^\k_{i,2}=0$, which implies the claim. 
According to type 7, since $z_{i+1}^\k=1$, gate $i$ is in the \grr.

We first show that in this regime, we must have $\alpha^\k_{i,1}=\alpha^\k_{i,2}=1$.
It is immediate that if $I_j(g_i^{\k})=1$ then $\alpha^\k_{i,j}=1$ for $j\in\{1,2\}$.
We now show that $I_j(g_i^{\k})=0$ implies $\alpha^\k_{i,j}=1$ for $j\in\{1,2\}$.
Suppose, without loss of generality, that $j=1$. Suppose $I_1(g_i^{\k})=0$ and for the sake of contradiction that $\alpha^\k_{i,1}=0$. Since $\alpha^\k_{i,1}$ is biased to $1$ it can only be $0$ if $\lambda^\k_{i,1}=1$. 
Since $\lambda^\k_{i,1}$ is biased to $0$, it can only be $1$ if $\beta^\k_{i,1}=1$.
Since $\beta^\k_{i,1}$ is biased to $0$, it can only be $1$ if $g_i=1$ and $\gamma^\k_{i,1}=0$.
However, since $\beta^\k_{i,1}=1$ and $g_i^\k=1$ and $\gamma^\k_{i,1}$ is biased to $1$, we have $\gamma^\k_{i,1}=1$.
Thus $\alpha^\k_{i,1}=1$.

We are left to show that in the \grr we must have $\delta^\k_{i,1}=\delta^\k_{i,2}=0$.
If $(I_1(g_i^\k),I_2(g_i^\k))=(0,0)$, it is immediate that $\delta^\k_{i,1}=\delta^\k_{i,2}=0$.
Suppose $(I_1(g_i^\k),I_2(g_i^\k))=(0,1)$. Then it is immediate that $\delta^\k_{i,1}=0$, and then since $\delta^\k_{i,2}$ is biased to $0$, it must also be $0$.  
The case $(I_1(g_i^\k),I_2(g_i^\k))=(1,0)$ is symmetric.
Finally, suppose $(I_1(g_i^\k),I_2(g_i^\k))=(1,1)$.
If $\delta^\k_{i,1}=\delta^\k_{i,2}=1$ then it is immediate that $\beta^\k_{i,3}=1$.
Then $\delta^\k_{i,1}$ and $\delta^\k_{i,2}$ are both indifferent to the edges of type 1, but they are not stable since they are biased to $0$.
Suppose $\delta^\k_{i,1}=1$ and $\delta^\k_{i,2}=0$.
If $\beta^\k_{i,3}=0$ then $\delta^\k_{i,1}$ is indifferent to the edges of type 1, but is biased to $0$ and hence is not stable.
If $\beta^\k_{i,3}=1$, then since $\beta^\k_{i,3}$ is biased to $0$, we must have $g_i=0$ and $\gamma^\k_{i,3}=1$.
But then $\gamma^\k_{i,3}$ is not stable since it is biased to $0$. 
The case $\delta^\k_{i,1}=0$ and $\delta^\k_{i,2}=1$ is symmetric.
%
\end{proof}

\begin{lemma}
\label{l:last_correct_gate}
Suppose $z_{i+1}^\k=0$ and $z_i^\k=1$ for some index $1\le i \le N$. \\
(a)
If $g_i^\k$ is correct then $\gamma^\k_{i,1}=\gamma^\k_{i,2}=0$ and $\gamma^\k_{i,3}=1$.\\
(b)
If $g_i^\k$ is not correct then $g_i^\k$ is indifferent w.r.t. edges of type 1 
but w.r.t. the edges only in type 8 deviating would improve her happiness.
\end{lemma}

\begin{proof}
According to type 7, since $z^\k_{i+1}=0$, gate $i$ is in the \gcr.
Thus  $\gamma^\k_{i,1}$ and $\gamma^\k_{i,2}$ are biased to $0$.
First suppose the gate is correct. 

If the correct output is $0$ then we have 
$g_i^\k=0$, and $z_i^\k=1$ by assumption.
Then $\gamma^\k_{i,1}$ and $\gamma^\k_{i,2}$ either prefer $0$ or are indifferent w.r.t. the edges in type 1 (depending on the values of $\beta^\k_{i,1}$ and $\beta^\k_{i,2}$). 
As they are biased to $0$ they will be $0$.
Suppose that $\gamma^\k_{i,3}=0$ for the sake of contradiction.
Then since it is biased to $1$, we must have $\beta^\k_{i,3}=0$.
Since $\beta^\k_{i,3}$ is biased to $1$, we must have 
$\delta^\k_{i,1}=\delta^\k_{i,2}=0$.
However, as the correct output is $0$, at least one of the input bits must be $1$. 
Suppose w.l.o.g. that $I_1(g_i^\k)=1$.
Then, since $\delta^\k_{i,1}$ is indifferent w.r.t. the edges in type 1 and is biased to $1$, it must be $1$, a contradiction.

Now suppose the correct output is $1$.
Thus we have $g_i^\k=1$, and $z_i^\k=1$ by assumption.
Then $\gamma^\k_{i,3}$ either prefers $1$ or is indifferent w.r.t. the edges in type 1 (depending on the value of $\beta^\k_{i,3})$. 
Since $\gamma^\k_{i,3}$ is biased to $1$, it will be $1$.
Suppose that $\gamma^\k_{i,1}=1$ for the sake of contradiction.
Since $\gamma^\k_{i,1}$ is biased to $0$ it can only be $1$ if $\beta^\k_{i,1}=0$.
Since $\beta^\k_{i,1}$ is biased to $1$ is can only be $0$ if  $\lambda^\k_{i,1}=0$.
Since $\lambda^\k_{i,1}$ is biased to $1$ it can only be $0$ if $\alpha^\k_{i,1}=1$.
Since the output is $1$ the input $I_1(g_i^\k)=0$, and then since $\alpha^\k_{i,1}$ is indifferent 
w.r.t. the edges in type 1 and is biased to $0$, it must be $0$, a contradiction.  
The same reasoning applies for $\gamma^\k_{i,2}$.
This completes the proof of (a).

Now suppose the output is incorrect.
Note that $g_i$ is indifferent w.r.t. the edges of type 1 if and only if 
$\beta^\k_{i,1}\ne\gamma^\k_{i,1}$ and $\beta^\k_{i,2}\ne\gamma^\k_{i,2}$ and $\beta^\k_{i,3}=\gamma^\k_{i,3}$.

First suppose the output is $0$.
Thus we have $g_i^\k=0$, and $z_i^\k=1$ by assumption.
Since the output is $0$ and incorrect, we have $I_1(g_i^\k)=I_2(g_i^\k)=0$, 

Suppose $\beta^\k_{i,1}=\gamma^\k_{i,1}=1$.
Then $\gamma^\k_{i,1}$ is indifferent 
w.r.t. the edges in type 1 and is biased to $0$, a contradiction. 
Now suppose $\beta^\k_{i,1}=\gamma^\k_{i,1}=0$.
Since $\beta^\k_{i,1}$ is biased to $1$, we have $\lambda^\k_{i,1}=0$.
Since $\lambda^\k_{i,1}$ is biased to $0$, we have $\alpha^\k_{i,1}=1$.
But $\alpha^\k_{i,1}$ is indifferent w.r.t. the edges of type 1 and biased to $0$, a contradiction.
Thus $\beta^\k_{i,1}\ne\gamma^\k_{i,1}$ and likewise $\beta^\k_{i,2}\ne\gamma^\k_{i,2}$.

Now suppose $\beta^\k_{i,3}=1$ and $\gamma^\k_{i,3}=0$.
Then $\gamma^\k_{i,3}$ is indifferent w.r.t. the edges in type 1 and biased to $1$, a contradiction.
Now suppose $\beta^\k_{i,3}=0$ and $\gamma^\k_{i,3}=1$.
Then $\gamma^\k_{i,3}$ prefers to be $0$ than $1$, a contradiction.
We have shown that $g_i^\k$ is indifferent w.r.t. edges of type 1 when 
the output is $0$ and incorrect.

Since $I_1(g_i^\k)=I_2(g_i^\k)=0$, and $\alpha^\k_{i,1}$ and $\alpha^\k_{i,2}$ are biased to $0$, we 
have  $\alpha^\k_{i,1}=\alpha^\k_{i,2}=0$.
Thus type 8 biases $g_i^\k$ to $1$, and it would gain by flipping as claimed.   

Now suppose the output is $1$.
Thus we have $g_i^\k=1$, and $z_i^\k=1$ by assumption.
%
Suppose $\beta^\k_{i,1}=\gamma^\k_{i,1}=0$.
Then $\gamma^\k_{i,1}$ prefers to be $1$ than $0$, a contradiction.
Suppose $\beta^\k_{i,1}=\gamma^\k_{i,1}=1$.
Then $\gamma^\k_{i,1}$ is indifferent w.r.t. edges of type 1 and is biased to $0$, a contradiction.
Thus $\beta^\k_{i,1}\ne\gamma^\k_{i,1}$ and likewise $\beta^\k_{i,2}\ne\gamma^\k_{i,2}$.

Now suppose $\beta^\k_{i,3}=1$ and $\gamma^\k_{i,3}=0$.
Then $\gamma^\k_{i,3}$ prefers to be $1$ than $0$, a contradiction.
Now suppose $\beta^\k_{i,3}=0$ and $\gamma^\k_{i,3}=1$.
Since $\beta^\k_{i,3}$ is biased to $1$, we must have 
$\delta^\k_{i,1}=\delta^\k_{i,2}=0$.
Since the output is $1$ and incorrect, we have at least one of $I_1(g_i^\k)$ and $I_2(g_i^\k)$ equal to $1$.
Suppose w.l.o.g. that $I_1(g_i^\k)=1$.
Then, since $\delta^\k_{i,1}$ is indifferent w.r.t. the edges in type 1 and is biased to $1$, it must be $1$, a contradiction.
We have shown that $g_i^\k$ is indifferent w.r.t. edges of type 1 when 
the output is $1$ and incorrect.

At least one of  $I_1(g_i^\k)$ and $I_2(g_i^\k)$ are $1$.
Suppose w.l.o.g. that $I_1(g_i^\k)=1$.
Then $\alpha^\k_{i,1}=1$.
Thus type 8 biases $g_i^\k$ to $0$, and it would gain by flipping as claimed.   
This completes the proof of (b).
\end{proof}

\begin{lemma}
\label{l:outputs_indiff_comparator}
If $d^\k=1$ and $d^{\nk}=0$ then for all $1\le i\le 2n$, node $g_i^\k$ is indifferent w.r.t. edges in type 4.
\end{lemma}

\begin{proof}
Each of these nodes is incident to exactly two type 4 edges both having the same weight.
For $1\le i\le n$ these are $(1, g_i^\k)$ and $(d^{\nk}, g_i^\k)$, while for   $n+1\le i\le 2n$ these are $(0, g_i^\k)$ and $(d^{\k}, g_i^\k)$.
The claim follows since $d^\k=1$ and $d^{\nk}=0$.
\end{proof}

\begin{lemma}
\label{l:allgatesfixed}
Suppose $d^\k=1$ and $d^{\nk}=0$. 
\begin{itemize}
\item[(a)] If $z^\k$ is in the \ccr then $z_i^\k=0$ for all $0\le i \le N+1$ and $y^\k=1.$
\item[(b)] If $z^\k$ is in the \crr then $z_i^\k=1$ for all $0\le i \le N+1$ and $y^\k=0.$
\end{itemize}
\end{lemma}
%
%
\begin{proof}
We start proving part (a). Since $z^\k$ is in the \ccr we immediately get $z_{N+1}^\k=0$. 
Assume, by way of contradiction, that there exists an index $1\le i\le N$ such that 
$z_i^\k=1$ and $z_{i+1}^\k=0$. 
First assume $g_i^\k$ is correct. 
Then $\gamma_{i,1}^\k=\gamma_{i,2}^\k=0$ and $\gamma_{i,3}^\k=1$ by \lref{l:last_correct_gate}(a). 
This and the fact that $z_{i+1}^\k$ is biased to $0$ implies that $z_i^\k=0$, a contradiction.
Now assume $g_i^\k$ is not correct. 
Then, by \lref{l:last_correct_gate}(b),  $g_i^\k$  is indifferent w.r.t. edges of type 1  
but w.r.t. the edges only in type 8 flipping would improve her happiness. If $3n+1\le i\le N$ 
this already implies that $g_i^\k$ would gain by switching, a contradiction. 
For the outputs and negated outputs, i.e., $1\le i\le 2n$, we know by \lref{l:outputs_indiff_comparator} 
that $g_i^\k$ is indifferent w.r.t. edges in type 4.
Moreover for the gates that represent the better neighbouring solution, i.e. $2n+1 \le i \le 3n$, 
we know by \lref{l:copy_better}(a) that $g_i^\k$ is indifferent w.r.t. edges in type 5.
(Recall that $w_i^\k$ is just another name for $g_{2n+i}^\k$.)
In either case, $g_i^\k$ would gain by switching, a contradiction. 
Thus, $z_i^\k=0$ for all $1\le i\le N$. It remains to show that $z_0^\k=0$ and $y^\k=1$.
Since $z_1^\k=0$ and $z^\k$ is in the \ccr, we get (by inspection of type 3 edges) that 
$z_0^\k=0$ which then implies $y^\k=1$.
This completes the proof of part (a).

To see part (b), observe that $d^{\nk}=0$ together with the bias of $y^\k$ to $0$ implies $y^\k=0$.
Now since $y^\k=0$ and $d^\k=1$ the bias of $z^\k_0$ enforces $z^\k_0=1$. The rest is by induction
since $z^\k_i=1$ and the bias directly implies $z^\k_{i+1}=1$ for all $0\le i \le N$. This completes the proof of part (b).
\end{proof}

%

\medskip
We now continue with the proof of Theorem~\ref{t:main_pls}.
Suppose we are in a Nash-stable outcome of the party affiliation game. 
For our proof we assume $C(v^0)\ge C(v^1)$. We will point out the small differences of the other case afterwards.   
Since $C(v^0)\ge C(v^1)$, $z^0$ is in the \ccr, i.e., all $z_i^0$ are biased to $0$ and $y^0$ is biased to 1 (by type 6). Thus, $z_{N+1}^0 = 0$.

The remainder of the proof splits depending on the coalition of  $z_1^0$ and $z_1^1$. By \lref{l:gate_incorrect} we
know that $z_1^\k=0$ implies that all gates in $C^\k$ are correct.


\noindent\underline{$z_1^0=1$:} 
By  \lref{l:gate_incorrect} we have $z_0^0=1$ and $y^0=0$. 
If $d^0=d^1=0$ then $d^0$ is better off changing to $1$ (by inspection of type 3 edges). 
If  $d^0=1$ then \lref{l:allgatesfixed}(a) implies $z_1^0=0$, a  contradiction. 
If $d^1=1$ and $z^1$ is in the \crr then by \lref{l:allgatesfixed}(b) and \lref{l:inputs_indifferent}, 
$v^1$ is indifferent w.r.t. type 1 edges. Thus $v^1=\eta^0$. But then either condition (ii) or (iii)
for putting $z^1$ in the \ccr (cf. type 6) are fulfilled. So $z^1$ has to be in the
\ccr.  \lref{l:allgatesfixed}(a) then implies $z_1^1=0$. But then the neighbourhood of $d^1$ in type 3 is dominated by $0$,
a contradiction to $d^1=1$. 

\noindent\underline{$z_1^0=0$ and $z_1^1=1$:} 
By \lref{l:gate_incorrect} we have $z_0^1=1$ and $y^1=0$. 
Since $C(v^0)\ge C(v^1)$ we know that $z^0$ is in the \ccr. 
So $z_1^0=0$ enforces $z_0^0=0$ and $y^0=0$.
By inspection of type 3 edges we have $d^0=0$ and thus $d^1=1$. 
First assume that $z^1$ is in the \crr, then $z_i^1=1$ for all $0\le i\le N+1$ and
\lref{l:inputs_indifferent} says that the inputs of all gates $g_i^1$ are indifferent
w.r.t. type 1 edges. In particular this holds for $v^1=(v_i^1)_{i\in[n]}$, so  $v^1= \eta^0$.
By \lref{l:copy_better}(b), $\eta^0= w^0$. Since $z_1^0=0$, $C^0$ is computing correctly and thus $ w^0=w(v^0)$.
Combining this we get $v_1=w(v^0)$ which contradicts our assumption that $z^1$ is in the \crr.
Thus $z^1$ is in the \ccr. Since $d^1=1$ we can apply  \lref{l:allgatesfixed}(a) to conclude $z_1^1=0$, a contradiction.

\noindent\underline{$z_1^0=0$ and $z_1^1=0$:} 
By  \lref{l:gate_incorrect} we have $z_0^0=z_0^1=0$ and $y^0=y^1=1$. 
Moreover we know that both circuits are computing correctly. 
If $d^0=1$ then $d^1=0$ and $d^0$ is indifferent w.r.t. type 3 edges. 
Since both circuits are computing correctly and $C(v^0)\ge C(v^1)$, the type~4 edges enforce $d^0=0$. 
But then $d^1$   is indifferent w.r.t. type 3 edges and the type 4 edges enforce $d^1=1$.
So, $d^0=0$ and $d^1=1$. If $z^1$ is in the \crr then \lref{l:allgatesfixed}(b) gives $z_1^1=1$, a contradiction.
Thus, $z^1$ is in the \ccr. Since $d^1=1$ we can apply \lref{l:copy_better}(b). This and the fact that $C^0$ is computing correctly 
implies $\eta^0=w(v^0)$. So $z^1$ can only be in the \ccr if $v^1=w(v^0)$. Since $C(v^0)\ge C(v^1)$ this implies that $v^0=v^1$ is a local optimum 
for the circuit $C$.

This finishes the proof in case $C(v^0)\ge C(v^1)$. The case $C(v^0) < C(v^1)$ is completely symmetric except here the conclusion $v^0=v^1$ in
the very last sentence leads to the contradiction $C(v^0) < C(v^0)$. So this case can't happen in a local optimum.

Note that throughout the construction we made sure that no node is incident to more than one negative edge. 
This completes the proof of \tref{t:main_pls}. 
\end{proof}
The instance produced by this reduction has the property that no node is indifferent between the 
two coalitions. We will make use of this property later in the paper.
\begin{corollary}
\label{c:oepastar}
\oepa is \PLS-complete even if restricted to instances where no player is ever indifferent between the two coalitions, i.e. \nioepa is \PLS-complete.
\end{corollary}

\section{Individual stability}
\label{sec:is}

In this section, we study the computational complexity of finding individual
stable outcomes.
We first provide a polynomial-time algorithm for \isTwo, which we define as the
problem of finding an individual stable outcome when only two coalitions can form,
i.e., we restrict the number of coalitions in the problem definition as for 
\pa.
The main result in this section is that \is is \plsc. We reduce from
\oepa that was shown to be \plsc in the previous section.
Our reduction uses exactly 5 coalitions, a restriction which we enforce using
supernodes.
We leave open the computational complexity of \isThree and~\isFour.

\begin{proposition}
  \label{p:is2}
  \isTwo can be solved in polynomial time.
\end{proposition}
\begin{proof}
We assume that there is at least one negative edge. Otherwise, the grand coalition is Nash-stable.
The algorithm goes as follows:
\begin{quote}
Start with any bipartition. 
Move nodes with incident negative edges so that they have a negative edge to the other coalition. 
In each of the two coalitions, contract all nodes with negative incident edges into a single node 
and call the contracted nodes $s$ and $t$. 
For any other node the new edge weights to $s$ and $t$ are the sum of the original edge weights to
the corresponding contracted nodes.  
Now (ignoring all edges between $s$ and $t$) compute a min cut between $s$ and $t$ 
via a max flow algorithm and assign the nodes accordingly.
\end{quote}
After the first stage, all nodes that we are about to contract have a negative
edge to the other coalition.  So they are not allowed to join the other
coalition. This property is preserved by contraction.  Afterwards, the flow
algorithm operates only on positive edges and computes a global minimum cut
between $s$ and $t$.  Thus, the cut also maximizes the total happiness of all
non-contracted nodes, so none of these nodes has an incentive to switch
coalitions.  All performed steps of the algorithm can be done in polynomial
time.
\end{proof}

Next we show that \is is \plsc\footnote{
The version of Theorem~\ref{thm:is} that appeared in~\cite{GS11} missed a 
special case that is dealt with 
here.} when we do not impose a restriction on the 
number of coalitions in the problem definitions as we did for \isTwo.

\begin{theorem}
\label{thm:is}
\is is \PLS-complete.
\end{theorem}

\begin{proof}
We start with an instance of \nioepa.
The instance has the property that no player is
ever indifferent between the two coalitions that make up
stable outcomes.
We add five supernodes 
which are connected by a complete graph of sufficiently large negative edges.
This enforces that in any stable outcome the supernodes are in different coalitions, say 
$0$, $1$, $2$, $3$, $4$. 
The supernodes are used to restrict which coalition a node
can be in in a stable outcome.
This is achieved by having large positive edges of equal weight to the corresponding supernodes.
All original nodes of the \nioepa instance are restricted to be $0$ or $1$.

We now show how to simulate a negative edge of \nioepa by an $\is$-gadget.
To do so, we replace a negative edge $(a,b)$ of weight $-w$ with the gadget in \fref{fig:is_gadget}. 
Nodes $a$ and $b$ are original nodes and restricted to $\{0,1\}$, 
node $a'$ is restricted to $\{0,1,2\}$, 
node $b'$ is restricted to $\{0,1,3\}$, 
node $c$ is restricted to $\{2,3,4\}$,
node $d$ is restricted to $\{2,4\}$,
node $e$ is restricted to $\{1,2\}$,
node $f$ is restricted to $\{0,2\}$.
As depicted in the gadget, nodes $b',c,d,e,$ and $f$ have additional offsets. 

Coalitions $2,3$ and $4$ are only used locally within the gadget.
\begin{figure*}
\centering
\resizebox*{\textwidth}{!}{%
\begin{minipage}[T]{0.71\linewidth}
\begin{center}
\resizebox*{\textwidth}{!}{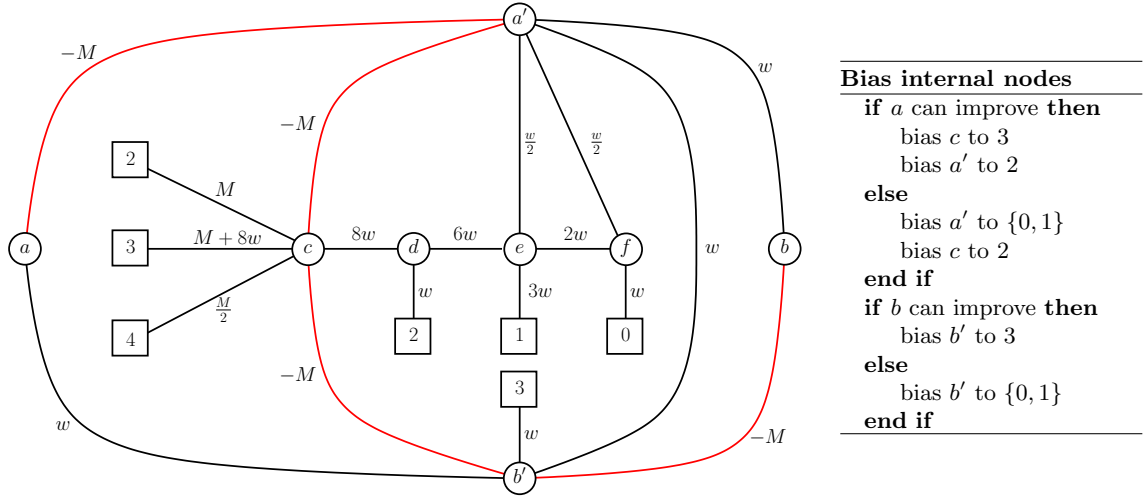}
\end{center}
\end{minipage}
\begin{minipage}[T]{0.02\linewidth}
\null~
\end{minipage}
\begin{minipage}[T]{0.27\linewidth}
\noindent
\hrule\smallskip
\textbf{Bias internal nodes}
\smallskip
\hrule 
\begin{algorithmic}
\If{$a$ can improve}
\State bias $c$ to $3$
\State bias $a'$ to $2$
\Else
\State bias $a'$ to $\{0,1\}$ 
\State bias $c$ to $2$
\EndIf
\If{$b$ can improve}
\State bias $b'$ to $3$
\Else
\State bias $b'$ to $\{0,1\}$
\EndIf
\end{algorithmic}
\hrule
\end{minipage}
} 
\caption{Gadget to replace negative edges}
\label{fig:is_gadget}
\end{figure*}
The pseudocode next to the gadget describes how the internal nodes of the gadget are biased. 
Here, checking whether a node can improve is  w.r.t. her \emph{original} neighborhood.
We use ``look at'' and ``bias'' as defined in
the following lemma and definition, which are analogous to those in \cite{ET11,MT10}. 
In particular, we check if a node can improve by looking at all nodes 
in her original neighborhood. 
\begin{lemma}
\label{l:compute_function}
For any  polynomial-time computable function $f:\{0,1\}^k\mapsto \{0,1,2,3\}^m$ one can construct a graph $G_f=(V_f,E_f, w)$ having the following properties:
(i) there exist $s_1, \ldots, s_k$, $t_1, \dots, t_m \in V_f$,
(ii) all edges $e\in E_f$ are positive,
(iii) $f(s_1, \ldots, s_k)=(t_1, \dots, t_m)$ in any 
stable solution of the 
hedonic game defined by $G_f$. 
\end{lemma}
\begin{proof}
We first show how to construct a graph $G_{f'}$, which implements a function $f':\{0,1\}^k\mapsto \{0,1\}^{4m}$. 
This part is similar to the proof of \lref{l:compute_function(old)}.
Afterwards, we show how to augment $G_{f'}$ to implement $f:\{0,1\}^k\mapsto \{0,1,2,3\}^m$
 
It is well known that for any polynomial computable function $f':\{0,1\}^k\mapsto \{0,1\}^{4m}$
one can construct a circuit $C$ with polynomial many gates that implements this function \cite[Theorem 9.30]{Sip06}. 
Clearly, we can also restrict $C$ to NOR gates with fan-in and fan-out at most 2. 
Organize the gates in levels according to their distance to $C$'s output; output gates of $C$ are at level 1.   

We replace each gate $g_i$ at level $\ell$ with the gadget in Figure~\ref{f:nor_is}. 
Nodes $u,v$ are inputs and $y$ is the output of the gate. Nodes $u,v,y$ are restricted to $\{0,1\}$.
Nodes $w$ and $x$ are internal to the gate and restricted to $w\in\{1,2\}$ and $x\in\{0,2\}$, respectively.
By construction of the NOR gate, we have that in any Nash-stable solution, $y=1$ if and only if $u=v=0$. In 
other words $y=NOR(u,v)$.

\begin{figure}[h]
\centering
\resizebox*{0.5\textwidth}{!}{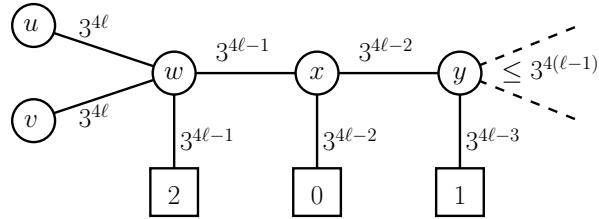}
\caption{NOR gate without negative edges}
\label{f:nor_is}
\end{figure}

If $u$ (or $v$) is an input of the circuit then we connect $u$ to the
corresponding input $s$-node by an edge of weight $3^{4\ell+1}$.
To connect the output $t$-nodes, we need to augment $G_{f'}$ in order to allow for the extended range of the function 
$f:\{0,1\}^k\mapsto \{0,1,2,3\}^m$. 
For each output node $t_i$ we will use 4 outputs of $G_{f'}$.

For each of this 4 outputs, we first change the domain by using a slightly modified NOR-gate.
Observe, that by changing the offsets and restrictions of nodes $w$, $x$, and $y$ in a NOR-gate, we can change the domain 
of the NOR gate to any two distinct values in $\{0,1,2,3,4\}$. 
E.g., if we want $y\in\{2,4\}$ and $y=2$ if and only if $u=v=0$, then we can
change the restrictions $w\in\{1,3\}$, $x\in\{2,3\}$, and $y\in\{2,4\}$; and the
offset of $w$ to 3, $x$ to 2, and $y$ to 4. Using this idea, we first change the
domain of the 4 outputs of $C$ to 
$\{0,4\}$, $\{1,4\}$, $\{2,4\}$, and $\{3,4\}$, respectively. 
The 4 modified outputs are then all connected to the output node $t_i$ (which we restrict to $\{0,1,2,3\}$) 
with edges of weight $1$. 
By the right choice of $f'$, in particular by ensuring that $f'$ forces exactly one of the modified outputs being 
$\neq 4$, we can implement any function~$f$.

The claim follows since our construction fulfils properties (i), (ii) and (iii).
\end{proof}


\begin{definition}
\label{d:function}
For a  polynomial-time computable function $f:\{0,1\}^k\mapsto \{0,1,2,3\}^m$ we
say that $G_f$ as constructed in \lref{l:compute_function} is a graph that
\emph{looks} at $s_1, \dots, s_k \in V_f$ and \emph{biases} $t_1, \dots, t_m \in
V_f$ according to the function $f$.
\end{definition}

Recall that
the instance of \nioepa has the property that no player is
ever indifferent between the two coalitions that make up
stable outcomes. By scaling edge weights we can implement the ``look at''
required to bias the internal nodes of the gadget without affecting their
original
preferences.

We say that node $a$ is \emph{locked} by the gadget if $a=1$ and $a'=0$ or  $a=0$ and $a'=1$.
Node $b$ is said to be locked accordingly.
The following three lemmas describe the operation of the gadget. All three lemmas should be read with the implicit clause:
\emph{If the internal nodes ($a'$, $b'$, $c$, $d$, $e$, $f$) of Figure~\ref{fig:is_gadget} are stable.} Let $\neg u$ denote the complement of $u$ over $\{0,1\}$.
\begin{lemma}
Node $c$ is either in coalition $2$ or $3$, while nodes $d$, $e$ and $f$ are in coalition~$2$.
\end{lemma}
\begin{proof}
We start by showing that $c\in\{2,3\}$. By way of contradiction assume $c=4$. Then 
$c=4 \Rightarrow d=4 \Rightarrow e=1 \Rightarrow f=0$. 
Now $a'$ strictly prefers coalitions $0$ and $1$ to coalition $2$ and is only blocked 
by $a$ from entering one of those coalitions. Thus $a'\in\{0,1\}$. 
Together with $b'\in\{0,1,3\}$ this directly implies that $c$ can improve by choosing $c=2$, 
contradicting our assumption. Thus $c\in\{2,3\}$.

From $c\neq 4$ we immediately get $d=2 \Rightarrow e=2 \Rightarrow f=2$, which completes the proof the lemma.
\end{proof}

\begin{lemma}
\label{l:noimprove}
If neither $a$ nor $b$ can improve then $a$ and $b$ are locked by the gadget.
\end{lemma}
\begin{proof}
Since neither $a$ nor $b$ can improve, $a'$ and $b'$ are biased to $\{0,1\}$ and $c$ is
biased to $2$.
If $c=2$ then the bias on $a'$ assures $a'=\neg{a}$. So $b'$ has an edge of weight $w$ to both
$0$ and $1$. Together with the bias this implies $b'=\neg{b}$.
If $c=3$ then the bias on $b'$ assures $b'=\neg{b}$. So $a'$ has an edge of weight $w$ to both
$0$ and $1$. Together with the bias this implies $a'=\neg{a}$.
So in both cases $a'=\neg{a}$ and $b'=\neg{b}$. 
The claim follows.
\end{proof}

\begin{lemma}
\label{l:improve}
If $a$ or $b$ (or both) can improve then one improving node is not locked while the other node is locked by the gadget. 
Moreover, if $a$ (resp. $b$) is not locked by the gadget then $b'=\neg{b}$ (resp. $a'=\neg{a}$).
\end{lemma}
\begin{proof}
We consider three cases: 
(i) only $a$ can improve, 
(ii) only $b$ can improve, 
(iii) $a$ and $b$ can improve.  

\noindent
\underline{Case (i) (only $a$):}
Here $c$ is biased to $3$, $a'$ is biased to $2$, and $b'$ is biased to $\{0,1\}$. 
First assume $c=2$. 
This enforces $a'=\neg{a}$ which together with the bias implies $b'=\neg{b}$. 
But then the bias on $c$ gives $c=3$, a contradiction. 
Thus $c=3$, which enforces $b'=\neg{b}$ and with the bias implies $a'=2$.
So $a$ is not locked and $b$ is locked.
\\
\underline{Case (ii) (only $b$):}
Here $c$ is biased to $2$, $a'$ is biased to $\{0,1\}$, and $b'$ is biased to $3$. 
First assume $c=3$. 
This enforces $b'=\neg{b}$ which together with the bias implies $a'=\neg{a}$. 
But then the bias on $c$ gives $c=2$, a contradiction. 
Thus $c=2$, which enforces $a'=\neg{a}$ and with the bias implies $b'=3$.
So $a$ is locked and $b$ is not locked.
\\
\underline{Case (iii) ($a$ and $b$):}
Here $c$ is biased to $3$, $a'$ is biased to $2$, and $b'$ is biased to $3$. 
If $c=2$ then this enforces $a'=\neg{a}$, which together with the bias implies $b'=3$. 
So in this case $a$ is locked and $b$ is not locked.
If $c=3$ then this enforces $b'=\neg{b}$, which together with the bias implies $a'=2$. 
So in this case $a$ is not locked and $b$ is locked.

In every case both claims of the lemma are fulfilled.
\end{proof}

To complete the proof we show that a stable outcome of the \is instance is also a stable 
outcome for the \nioepa instance.
Suppose the contrary. Then there must exist an original node which is stable for $\is$ but not for
\nioepa. Clearly such a node must be the node $a$ or $b$ for some gadget. So either $a$ or $b$ (or both)
can improve. But then by the first statement in  \lref{l:improve} one of the improving nodes is unlocked, say $a$. Since $a$ was only incident to one negative edge in the \nioepa instance, $a$ cannot be locked by any
other gadget. Moreover, by the second statement in \lref{l:improve}, $a$ is now connected in the gadget
by a positive edge to the node $b'$ and $b'=\neg{b}$. 
On the one hand, if $a=b$ then the original edge $(a,b)$ contributes $-w$ to $a$'s utility while now $a$ receives $0$ from the edge $(a,b')$. 
On the other hand, if $a\neq b$ then the corresponding utility contributions are $0$ and $w$.  
So if $a$ changes strategy then the difference in her utility w.r.t. $b$ is the same in both problems, since we just shifted the utility of node $a$ w.r.t. $b$ by $w$. 
So $a$ is also 
not stable for $\is$, a contradiction. This finishes the proof of \tref{thm:is}.
\end{proof}

\section{Other veto-based stability concepts}
\label{sec:other_veto}

In \is a single player can veto against others joining her coalition but there is no restriction on leaving a
coalition. The following proposition shows that adding certain leaving conditions yields polynomial-time
convergence from the all-singleton partition.
\begin{proposition}
\label{prop:existsneg}
Any problem in column~3 of \fref{f:table} can be solved in polynomial time 
provided that the leaving condition requires that the leaving node has at least one negative edge 
within the coalition. In particular this hold for the problems in cells 3B, 3C, and 3D. 
\end{proposition}

\begin{proof}
We use local improvements starting from the set of singleton coalitions.
Then a player can make at most one improving step, since all edges in resulting non-singleton coalitions
will be positive because of the veto-in restriction, and so no player can leave such a coalition.
Hence we arrive at a stable outcome in at most $|V|$ improving steps.
\end{proof}

Interestingly, requiring veto-out feasibility is already enough for polynomial-time convergence even if we have no restriction on the entering condition. This stands in contrast to \tref{thm:is}.
\begin{proposition}
\label{prop:justgrow}
All problems in row C of \fref{f:table} can be solved in polynomial time 
by local improvements using at most $2|V|$ improving steps.
\end{proposition}

\begin{proof}
To get a running time of $2|V|$ (rather than $O(|V|^2)$) we restrict players from joining a non-empty coalition to which they have no positive edge.
This ensures that whenever a player joins a non-empty coalition then this player (and all players to which she is
connected by a positive edge in the coalition) will never move again. Moreover, a player can only start a new coalition once.
It follows that each player can make at most two strategy changes. In total we have at most $2|V|$ local improvements.
\end{proof}

\section{\wcis}
\label{sec:sum_cis}


Next we study \wcis, where a deviating player's 
total weight to the new coalition is non-negative, and to
the old coalition is non-positive.
Even though deviations are very restricted here, 
it is \PLS-complete to compute a stable outcome. 
\begin{theorem}
\label{thm:wcis}
\wcis is \PLS-complete.
\end{theorem}
\begin{proof}
We reduce from \lmc. Consider an arbitrary instance of \lmc with only integer edge weights. 
Recall that such an instance can be cast as an instance of \pa by negating the weights of the edges. 
Let $G=(V,E,w)$ represent the \pa instance. 
For each player $i\in V$ let $\sigma_i$ be the total weight of edges incident to player $i$, i.e. $\sigma_i=\sum_{(i,j)\in E} w_{(i,j)}$. 
Observe that $\sigma_i$ is a negative integer.
We augment $G$ by introducing two new players, called \emph{supernodes}.  
Every player $i\in V$ has an edge of weight $\frac{-\sigma_i}{2}+\frac{1}{4}$ to each supernode.
The two supernodes are connected by an edge of weight $-M$ where $M$ is sufficiently large  (i.e., $M>\sum_{i\in V}(\frac{-\sigma_i}{2}+\frac{1}{4})$). 
The resulting graph $G'$ represents 
our \wcis instance.

Consider a stable outcome of the \wcis instance $G'$. 
By the choice of $M$ the two supernodes will be in different coalitions. 
Now consider any player $i\in V$. 
If $i$ is not in a coalition with one of the supernodes, then $i$'s payoff is negative. On the other hand
joining the coalition of one of the supernodes yields positive payoff, since $2(\frac{-\sigma_i}{2}+\frac{1}{4})+\sigma_i>0$.
Thus, each player $i\in V$ will be in a coalition with one of the supernodes. 
So our outcome partitions $V$ into two partitions, say $V_1, V_2$. 

It remains to show that any stable outcome for the \wcis instance 
is also a local optimum for the \pa instance.
Assume that the outcome of the \wcis instance is stable but in the corresponding outcome of \pa instance
there exists a player $i$ which can improve by joining the other coalition. W.l.o.g. assume $i\in V_1$.
Then,  
$
   \sum_{s\in V_1} w_{(i,s)}<\sum_{s\in V_2} w_{(i,s)}.
$
With $\sigma_i=\sum_{s\in V} w_{(i,s)}$ and since $\sigma_i$ is integer, we get
$$
  \sum_{s\in V_1} w_{(i,s)}\le\frac{\sigma_i}{2}-\frac{1}{2} <\frac{\sigma_i}{2}<\frac{\sigma_i}{2}+\frac{1}{2}\le\sum_{s\in V_2} w_{(i,s)}.
$$
It follows that in the \wcis instance, player $i$'s payoff is negative in her current coalition $V_1$ whereas joining $V_2$ would yield positive payoff.
This contradicts our assumption that we are in a stable
outcome of the \wcis instance.
The claim follows.
\end{proof}

\section{Voting-based deviations}
\label{sec:voting}

In this section we study the complexity of computing stable outcomes under various voting-based 
stability 
requirements.
We start by showing \PLS-hardness for the case that a deviating player needs a $T_{in}$ majority
in the target coalition but there is no restriction on leaving coalitions. 
\begin{theorem}
\label{thm:votein}
\votein is \PLS-complete for any voting
threshold $0\le T_{in}< 1$.
\end{theorem}
\begin{proof}
We reduce from \nioepa represented by an edge-weighted graph $G=(V,E,w)$.
Let $\Delta(G)$ be the maximum degree of a node in $G$. 
Recall that no player is ever indifferent between the two coalitions.

First observe that the case $T_{in}>\frac{\Delta(G)-1}{\Delta(G)}$ is exactly 
the same as \is (for which we show hardness in \tref{thm:is}), since in this case one negative edge is enough to veto a player joining a coalition.
In the following we assume $T_{in}\le\frac{\Delta(G)-1}{\Delta(G)}$.

We augment $G$ as follows: 

For every negative edge $(a,b)$ in $G$ we introduce $2\Delta(G)-2$ new nodes,
called \emph{followers}, and
connect them with $a$ and $b$ as shown in the \fref{f:votein}. Both, $a$ and $b$, get $\Delta(G)-1$
followers and have a $\delta$ edge to each of them. Moreover, the followers have also an edge
of weight $\varepsilon$ to the other node.
Here $0<\varepsilon<\delta$ and $\delta$ is small enough so that the player preferences of the original players ($a$ and $b$) are still determined only by the original edges.
\begin{figure}
\begin{center}
\resizebox*{0.5\textwidth}{!}{
\includegraphics{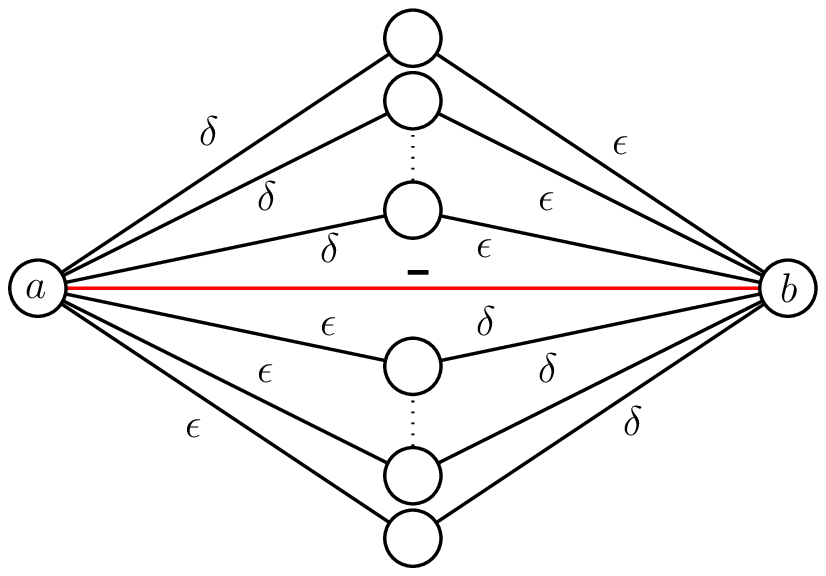}
}
\end{center}
\caption{Gadget used for showing that \votein is \plsc. The gadget augments negative edges with followers that ensure that there is always a $T_{in}$-majority when a player enters a coalition.} 
\label{f:votein}
\end{figure}
In a stable outcome the followers will be in the same coalition as their ``leader'', i.e., the node to 
which they have a $\delta$ edge.
The followers make sure that their is always a $T_{in}$-majority for entering a coalition. In other words, in a stable outcome
of the \votein instance, the voting doesn't impose any restrictions.

To ensure that any stable outcome for the \votein instance has only two coalitions we further augment $G$ by introducing two 
  new players, called \emph{supernodes}. 
  Every player $i\in V$ has an edge of weight $W>\sum_{e\in E} |w_e|$ to each of the supernodes. 
  The two supernodes
  are connected by an edge of weight $-M$, where $M> |V|\cdot W$. This enforces that the two 
  supernodes are in a different coalition in any stable outcome.
  Moreover,  by the choice of $W$, each player in $V$ will be in a coalition with one of the supernodes. 
 The fact that edges to supernodes have all the same weight directly implies that a stable outcome for the \votein instance 
 is also a stable outcome for the \nioepa instance.
The claim follows. 
\end{proof}

In contrast to \votein, \voteout is conceptually different. In \voteout a coalition of two players connected 
by a positive edge is vote-out stable. This makes it hard to restrict the number of coaltions. 
Doing this is probably the key for proving \PLS-hardness also for \voteout.
For the following theorem we consider a version of \voteout where the number of coalitions are restricted 
by the problem. 
Let $k$-\voteout be the problem of computing a vote-out stable outcome when at most $k$ coalitions
are allowed. Observe that for any $k\ge 2$ such a vote-out stable outcome exists and that
local improvements starting from any $k$-partition converge to such a stable outcome.
\begin{theorem}
\label{thm:voteout}
$k$-\voteout is \PLS-complete for any voting threshold $0\le T_{out}< 1$
and any $k\ge 2$.
\end{theorem}
\begin{proof}
Our reduction is from \oepa, but we first reduce to the intermediate problem \oens, which 
is a restricted version of \ns where each player is only incident to at most one negative edge.
Consider an instance of \oepa which is represented as an edge-weighted graph $G=(V,E,w)$. 
We augment $G$ with two supernodes in exactly the same way as in Theorem~\ref{thm:votein}.
This ensures that any stable outcome of the \oens instance uses only two coalitions and thus is also a stable outcome 
 for the \oepa instance.
Hence, \oens is \plsc.


We now reduce from \oens to $k$-\voteout. Let $G$ be the graph corresponding to an instance of  \oens.
Let $\Delta(G)$ be the maximum degree of a node in~$G$. 
%
We augment $G$ as follows: 
We introduce $s\cdot k\cdot \Delta(G)$ new nodes where $s$ is an integer satisfying 
$s\ge \frac{T_{out}}{1-T_{out}}$. 
Those nodes are organized in $s\cdot \Delta(G)$ complete 
graphs of $k$ nodes each. 
All the edges in the complete graphs have weight $-M$ where $M$ is sufficiently large ($M>|V|\cdot \Delta(G)\cdot\varepsilon$ will do). 
Moreover, we connect every original node $u\in V$ to every new node with an edge of weight $-\varepsilon$,
where $\varepsilon>0$. 

By the choice of $M$ and since at most $k$ coalitions are allowed, in any stable solution there will be one
node from each complete graph in each of the $k$ coalitions. 
This shifts the utility of each player $i\in V$ with respect to each 
coalition by $-s\cdot \Delta(G)\cdot\varepsilon$.
Moreover, every original node has at least $s\cdot \Delta(G)$ negative edges to each coalition. 
Since each node is incident to at most $\Delta(G)$ positive 
edges, it follows that the fraction of negative edges to each coalition is at least $\frac{s}{s+1}\ge T_{out}$.
Thus, in every stable outcome all nodes  $u\in V$ have a $T_{out}$-majority for leaving their coalition. 
This implies that in the corresponding outcome of the \oens instance, no player can improve her utility
by joining one of the $k$ coalitions used in $k$-\voteout. 
Moreover, in every stable outcome the utility of each node $u\in V$ with respect to the 
set of original nodes $V$ is non-negative, 
since $u$ has at most one negative incident edge in the \oens instance and $k\ge 2$. 
It follows that a stable outcome for the $k$-\voteout instance is also a stable outcome for the \oens instance.
The claim follows.
\end{proof}

It is an interesting open problem whether \PLS-completeness also holds if the
restriction on the number of allowed coalitions is dropped. Can we construct a
gadget that imposes this restriction without restricting the problem a priori? 

Since \votein and a restricted version of \voteout are \PLS-complete 
it's interesting to study the combination of both problems. 
What happens if we require vote-in stability \emph{and} vote-out stability? With a mild assumption on
the voting thresholds $T_{in}, T_{out}$, we establish:
\begin{theorem}
\label{thm:voteinout}
For any instance of \voteinout with voting thresholds 
$T_{in}, T_{out} \ge \frac{1}{2}$ and $T_{in}+T_{out}>1$, 
local improvements converge in $\bO(|E|)$ steps.
\end{theorem}
\begin{proof}
For any outcome $p$ define a potential function $\Phi(p)=\Phi^+(p)-\Phi^-(p)$, where $\Phi^+(p)$ (resp. $\Phi^-(p)$) is the number of
positive (resp. negative) internal edges, i.e. edges not crossing coalition boundaries. 
Consider a local improvement of some player $i$ from coalition $p(i)$ to $p'(i)$. 
Since $T_{out}\ge \frac{1}{2}$, player $i$ has at least as many negative as positive edges to $p(i)$. Likewise since $T_{in}\ge \frac{1}{2}$, 
player $i$ has at least as many positive as negative edges to $p'(i)$. So $\Phi(p)$ cannot decrease by a local improvement.
Moreover, since $T_{in}+T_{out}>1$, one of the threshold inequalities must be strict, which implies
$\Phi(p')>\Phi(p)$. The claim follows since $-|E|\le \Phi(p) \le |E|$ and  $\Phi(p)$ is integer.
\end{proof}
Without the assumption on the voting thresholds from Theorem~\ref{thm:voteinout}, 
the complexity of computing stable outcomes is 
an interesting open problem, in particular the case $T_{in}=T_{out}=1/2$.

\section{Open problems}
\label{sec:conc}

In this paper, we studied the computational complexity of 
finding stable outcomes in hedonic games.
We show that \ns is \plsc.
On the other hand we show that \cis, that is finding a stable
outcome where any member of a coalition can block (veto) a 
player from leaving or joining, can be solved in polynomial 
time.
For the case that a player can only block a player from joining,
we show that the corresponding problem \is
is \plsc (Theorem~\ref{thm:is}).
Our reduction to \is uses five coalitions.
On the other hand, \isTwo, where the number of coalitions is restricted to two,
is solvable in polynomial time (Proposition~\ref{p:is2}).
This leaves open the complexity of \isThree and \isFour, where
the number of coalitions is restricted to three or four, respectively.

We then study cases where members of a coalition can vote on whether  
to allow a player to leave or join a coalition.
The problem \votein is parameterized by a voting threshold, $T_{in} \in [0,1]$.
\is can be seen as \votein with $T_{in}=1$.
Theorem 5 shows hardness for $0 \le T_{in} < 1$, so in fact we show that \votein
is \plsc for all voting thresholds.
In contrast, we show that the case of \voteout with $T_{out}=1$ 
is polynomial-time solvable (Proposition~\ref{prop:justgrow}). 
This suggests that \voteout is conceptually different from \votein. 
Indeed, it seems difficult to restrict the coalitions in this case.
We do show that $k$-\voteout, where we restrict the outcome to have at most $k$ coalitions, 
is \plsc for $0\le T_{out}<1$, but we leave 
the complexity of \voteout as an interesting open problem.
 
On the positive side, we show that
local improvements converge in polynomial time in the case of requiring 
\emph{both} vote-in- and vote-out- 
stability with $T_{in},T_{out}\ge 0.5$ and $T_{in}+T_{out}>1$.
We leave open the interesting case of \voteinout with voting thresholds
that do not satisfy 
$T_{in}, T_{out} \ge \frac{1}{2}$ and $T_{in}+T_{out}>1$.
We also leave open the case of finding an outcome that is vote-in 
and sum-out stable.


\citet{ET11} showed that local max cut is \plsc even when the input graph
has degree at most five.
In contrast, \citet{Pol95} gives a polynomial-time
algorithm for graphs with degree at most three.
It would be interesting to study degree restrictions for additively-separable
hedonic games.

\bibliographystyle{abbrvnat}
\bibliography{bib}

\end{document}